\documentclass[12pt, reqno]{amsart}

\IfFileExists{mymtpro2.sty}{%
  \usepackage[subscriptcorrection]{mymtpro2}
}{}

\usepackage{a4, upgreek, amssymb, amsmath, mathrsfs}
\usepackage{appendix}
\usepackage{amssymb,amsmath, amsfonts, euscript, mathrsfs, mathabx}
\usepackage{scalerel,stackengine}

\newtheorem{theorem}{Theorem}[section]
\newtheorem{lemma}{Lemma}[section]
\newtheorem{corollary}{Corollary}[section]
\newtheorem{prop}{Proposition}[section]
\newtheorem{remark}[theorem]{Remark}

\newtheorem{definition}[theorem]{Definition}

\newcommand{\labelnummer}{\mbox{\normalfont (\roman{numcount})}}%

\makeatletter

  {\let\curlabelspeicher\@currentlabel%
    \begin{list}{\labelnummer}%
      {\usecounter{numcount}\leftmargin0pt%
        \topsep0.5ex\partopsep2ex\parsep0pt\itemsep0ex\@plus1\p@%
        \labelwidth2.5em\itemindent3.5em\labelsep1em%
      }%
    \let\saveitem\item%
    \def\item{\saveitem%
      \def\@currentlabel{{\upshape\curlabelspeicher}$\,$\labelnummer}}%
    \let\savelabel\label%
    \def\label##1{\savelabel{##1}%
      \@bsphack%
        \ifmmode\else%
          \protected@write\@auxout{}%
          {\string\newlabel{##1item}{{\labelnummer}{\thepage}}}%
        \fi%
      \@esphack%
    }%
  }{\end{list}}%

\renewcommand{\appendix}{\def\thesection{\textsc{Appendix}}}

 \let\leq\le
 \let\geq\ge

 \let\Im\undefined

\DeclareMathOperator{\Im}{Im}
\newcommand{\ud}{\mathrm{d}}

\DeclareMathOperator{\tr}{tr\kern1pt}

\makeatletter

\newif\ifper\pertrue
\def\per{.}

\def\bti{\@ifnextchar[\bbti\bbbti}
\def\bbti[#1]#2{#2, #1.}
\def\bbbti#1{#1.}

\def\z{\@ifnextchar[\zz\zzz}
\def\zz[#1]#2#3#4#5{\perfalse\emph{#2} \textbf{#3}, #4 (#5) [#1]}
\def\zzz#1#2#3#4{\emph{#1} \textbf{#2}, #3 (#4)\ifper\per\fi\pertrue}

\def\pub{\@ifstar\pubstar\pubnostar}
\def\pubnostar{\@ifnextchar[\@@pubnostar\@pubnostar}
\def\@@pubnostar[#1]#2#3#4{#2, #3, #4, #1\ifper\per\fi\pertrue}
\def\@pubnostar#1#2#3{#1, #2, #3\ifper\per\fi\pertrue}
\def\pubstar[#1]#2#3#4{\perfalse #2, #3, #4 [#1]\pertrue}

\makeatother

\newcommand{\bel}{\begin{equation} \label}
\newcommand{\ee}{\end{equation}}

\def\beq{\begin{equation}}
\def\eeq{\end{equation}}
\newcommand{\bea}{\begin{eqnarray}}
\newcommand{\eea}{\end{eqnarray}}
\newcommand{\beas}{\begin{eqnarray*}}
\newcommand{\eeas}{\end{eqnarray*}}

{

\newcommand{\B}{\mathbb{B}}

\newcommand{\R}{\mathbb{R}}

\newcommand{\Z}{\mathbb{Z}}

\newcommand{\C}{\mathbb{C}}
\newcommand{\E}{\mathbb{E}}

\newcommand{\Schr}{Schr\"odinger }

\begin{document}

\title[Modified Thouless formula for the Bethe lattice]{Ergodic Schr\"odinger operators on the Bethe lattice and a modified Thouless formula}

\author[P.\ D.\ Hislop]{Peter D.\ Hislop}
\address{Department of Mathematics,
    University of Kentucky,
    Lexington, Kentucky  40506-0027, USA}
\email{peter.hislop@uky.edu}

\author[C.\ A.\ Marx]{Christoph A.\ Marx}
\address{Department of Mathematics,
Oberlin College,
Oberlin, Ohio 44074, USA}
\email{cmarx@oberlin.de}



\begin{abstract}
The main result of this paper is a modified Thouless formula relating the density of states for ergodic  Schr\"odinger operators on the Bethe lattice to the Lyapunov exponent. The modified Thouless formula consists of a Thouless-like term, involving the density of states, and a remainder term. The remainder term vanishes when the connectivity 
$\kappa$ equals one, yielding the usual Thouless formula for ergodic  Schr\"odinger operators on $\Z$. We prove the remainder term is nontrivial for $\kappa \geq 2$. We also discuss the automorphism group of the Bethe lattice and its relation to ergodic Schr\"odinger operators. In particular, we clarify the use of the multiparameter noncommutative ergodic theorem in evaluating the limit of Green's functions along certain paths. 
\end{abstract}

\maketitle \thispagestyle{empty}

\tableofcontents

\vspace{.2in}



\section{Introduction}\label{sec:introduction1}
\setcounter{equation}{0}
 
We study ergodic \Schr operators on the Bethe lattice $\mathbb{B} = (\mathcal{V}, \mathcal{E})$ with vertices $\mathcal{V}$, edges $\mathcal{E}$, and connectivity $\kappa \geq 2$. The Bethe lattice (also known as the infinite Cayley tree) is a regular tree in which each each vertex is connected to $\kappa + 1$ neighboring vertices, where $\kappa \geq 2$ is called the {\em{connectivity}}. The constant degree $\kappa + 1$ of each vertex is also known as the {\em{coordination number}} of the Bethe lattice in recognition of the radial structure (coordination spheres or levels) that arises when designating one vertex as the root of the tree; see Figure \ref{fig_bethe-lattice_basic}. For the sake of consistency, we mention that we used the coordination number in our previous papers \cite{hislop_marx_1,hislop_marx_3} but decided to use the connectivity here because $\kappa$ is more directly related to the spectrum of the free Laplacian on $\mathbb{B}$, given by the closed interval $[-2 \sqrt{\kappa}, 2 \sqrt{\kappa}]$. Since the main result of this paper generalizes results for Schr\"odinger operators on $\mathbb{Z}$ to Schr\"odinger operators on $\mathbb{B}$, we note that $\mathbb{Z}$ can be viewed as a trivial case of the Bethe lattice if one admits $\kappa=1$. 
\begin{figure}
	\includegraphics[width = \textwidth]{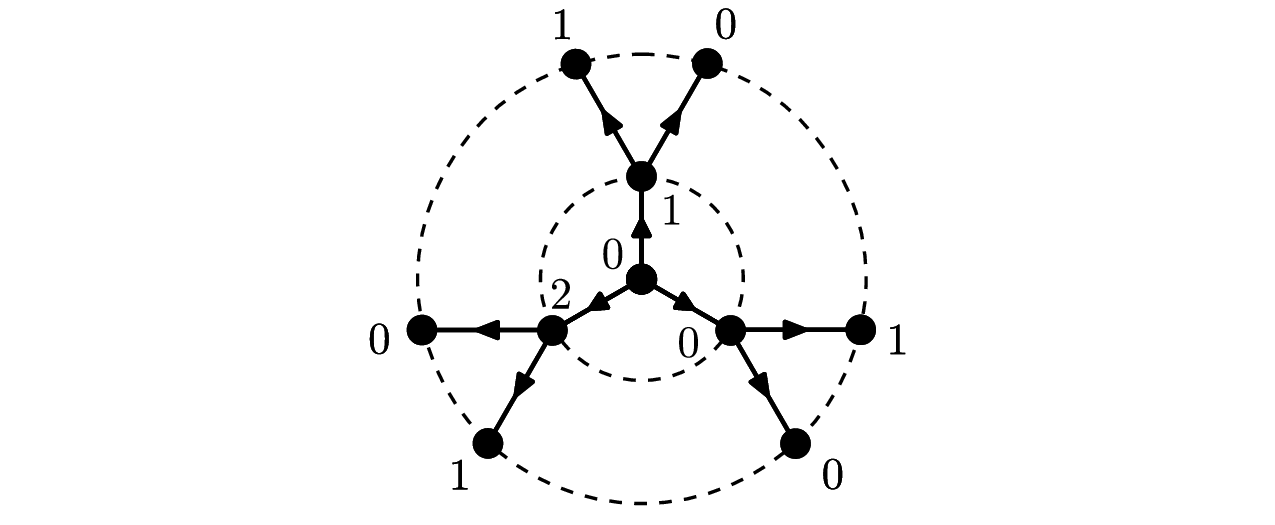} 
	\caption{Basic geometry of the Bethe Lattice: the figure shows a finite section of the Bethe lattice for connectivity $\kappa = 2$. Distinguishing one vertex as the root (shown in the center) results in a radial structure (coordination spheres or levels) about the root, where each coordination sphere consists of vertices at a fixed distance (=number of edges) from the root. The arrows display the natural partial order emanating from the root (see Section \ref{subsec:bethe1} for a precise definition); in particular, while the root is connected to $\kappa + 1$ forward neighbors, each other vertex has precisely $\kappa$ vertices in the forward direction. The labeling of the vertices is explained in Section \ref{subsec:bethe1}; see also Figure \ref{fig_bethe-lattice_vertex-labeling}.}
	\label{fig_bethe-lattice_basic}
\end{figure}

The Bethe lattice $\B$ is a regular tree graph so that each pair of distinct vertices is connected by a unique path; or equivalently, so that there are no closed loops. The tree structure leads to a recursive nature which is a feature shared with Schr\"odinger operators on $\mathbb{Z}$. In stark contrast to $\mathbb{Z}$, the Bethe lattice, however, exhibits hyperbolic geometry: finite graph approximations, obtained by restricting to a finite number of levels, are distinguished by the property that the number of vertices comprising the surface area is of the same order as the number of vertices in the truncated graph. This number grows exponentially like $\kappa^L$, where $L$ is the index of the coordination sphere. Hence, infinite volume limits are different from those for the cubic lattice $\Z^d$ where the ratio of surface area to volume vanishes as $L \to \infty$. Because of these properties, the Bethe lattice has served as a testing ground for random \Schr operators. One thinks of the Bethe lattice as ``quasi one-dimensional'' but, as a result of the expansive geometry, random Schr\"odinger operators on $\mathbb{B}$ exhibit spectral characteristics very different from their counterparts on the lattice $\Z$. 

This difference in the spectral properties already arises for the discrete Laplacian $\Delta_\mathbb{B}$ (i.e., for zero potential). There, the hyperbolic geometry leads to an unexpected spectral difference 
between the spectrum of infinite volume operator $\Delta_\mathbb{B}$ (defined in (\ref{eq_laplacian})) and its finite volume restrictions $\Delta_{\Lambda_L}$ to coordination spheres $\Lambda_L$ at distance $L$ from the root (see (\ref{eq_coordination-spheres_defn}) and (\ref{eq_defn_restriction}) for precise definitions). This difference manifests itself in the following strict inequality between the operator norms:
\begin{equation}
\Vert \Delta_{\Lambda_L} \Vert = \kappa + 1 > 2 \sqrt{\kappa} = \Vert \Delta_\mathbb{B} \Vert ~\mbox{.}
\end{equation}
We note that the inequality becomes an equality for the discrete Laplacian on $\mathbb{Z}$ (or $\kappa = 1$). For some references on earlier work on the Bethe lattice, we refer to the introductions of \cite{AcostaKlein, aw_weak-disorder} and to \cite{Abou-ChacraThoulessAnderson_JPhysC_1973, KS_83, Thouless_PRL_1977}.

For the case of random Schr\"odinger operators on $\mathbb{B}$ with identically distributed (iid) random potentials, absolutely continuous spectrum is known to persist for non-zero potentials (disorder) \cite{klein_ann1,klein_bethe1} (see also \cite{fhs_06, asw}). This is in stark contrast to the spectral behavior for random Schr\"odinger operators on $\mathbb{Z}$. Any non-zero disorder destroys the absolutely continuous spectrum and leads to spectral and dynamical localization throughout the spectrum  \cite[Chapter 10]{aw_book}. Concerning the weak disorder behavior of random Schr\"odinger operators on $\mathbb{B}$ with iid potentials, we mention the results by Acosta and Klein \cite{AcostaKlein} on the analyticity of the density of states for probability measures close to Cauchy, and results on the nature of the spectrum by Klein \cite{klein_bethe1}, with generalizations in \cite{KleinSadel}, which established that for weak disorder, absolutely continuous spectrum persists in the {\em{interior}} of the spectrum of the discrete Laplacian $E \in [-2 \sqrt{\kappa}, 2 \sqrt{\kappa}]$. In \cite{aw_JEMS_2013}, Aizenman and Warzel developed techniques ("resonant delocalization criterion'') which, contrary to earlier expectations, provided strong evidence for purely absolutely continuous spectrum across the entire spectrum for sufficiently weak disorder and bounded potentials. For unbounded potentials, their criterion for resonant delocalization implies the persistence of some absolutely continuous spectrum for $2 \sqrt{\kappa} < \vert E \vert < \kappa + 1$ in the weak disorder regime. 

A key quantity in the work of Aizenman and Warzel 
\cite[Section 4]{aw_JEMS_2013} is an extension of the concept of the Lyapunov exponent to random Schr\"odinger operators on the Bethe lattice, which they define as the negative expectation of the logarithm of a {\em{Weyl $m$-function}}:
\beq \label{eq_LE_defn_Weyl}
\mathcal{L}(z) = - \E \{ \log \vert G_{0;\omega} (z; y,y) \vert \} ~\mbox{.}
\eeq
Here, $G_{0;\omega} (z; y,y)$ represents one of the $\kappa+1$ equivalent Weyl $m$-functions for Schr\"odinger operators on $\mathbb{B}$: $G_{0;\omega} (z; y,y)$ is the diagonal element of the resolvent of the Schr\"odinger operator decoupled at its root $(0)$, 
where $y$ is any nearest neighbor of $(0)$.


\subsection{Main results}\label{subsec:main1}

The main result of this paper is the proof of a modified Thouless formula relating the density of states of a random \Schr operator on $\mathbb{B}$ to the Lyapunov exponent defined in \eqref{eq_LE_defn_Weyl}. 
This result is based on several new results on the structure of the Bethe lattice and ergodic \Schr operators on the Bethe lattice. 
We clarify the automorphisms of the Bethe lattice, their use in designating the unique path from the origin to any vertex of the Bethe lattice, 
and their role in building 
ergodic and random \Schr operators on $\mathbb{B}$.
We improve and extend a result of Acosta and Klein \cite{AcostaKlein} in 
Theorem \ref{thm_analytic-miracles} on the approximation of diagonal elements of the Green function by the Green function of the \Schr operator restricted to a finite subset. This uses the random walk expansion of the Green function described in Theorem \ref{prop_RW-exp}. We also prove results on the approximation for functions of the corresponding infinite-volume and finite-volume \Schr operators. 

Our next main result is that  $\mathcal{L}(z)$, defined in (\ref{eq_LE_defn_Weyl}), may be computed 
from the limit of the type
\beq
\dfrac{1}{L} \log \left\vert G_\omega(z; \gamma(0), \gamma(L-1) ) \right\vert \to - \mathcal{L}(z) ~\mbox{, as $L \to \infty$ .} 
\eeq
This requires a detailed analysis of the relation between the paths and the ergodic theorem. Due to the fact that the automorphisms are 
non-commuting, the existence of the limit does not follow from standard results. 

Finally, our main result combines this representation with the approximation of the Green function $ G_\omega(z; \gamma(0), \gamma(L-1) )$ by the Green function of finite-volume \Schr operators, we arrive at the modified Thouless formula
\beq\label{eq:lyap_thouless0}
\mathcal{L}(z) = \int_\Sigma \log | z - E| ~dn_v(E) + R(z) ,
\eeq
where $dn_v$ is the density of states measure and $R(z)$ is the nontrivial remainder for $\kappa \geq 2$. 


Formula \eqref{eq:lyap_thouless0} is a modification of the standard  Thouless formula familiar from ergodic Schr\"odinger operators on $\mathbb{Z}$, see \cite[chapter 12]{aw_book}, expressed by the presence of the remainder term $R(z)$. The remainder term, depending on $\kappa$, is a result of the averaged effects of the coupling of paths to their surroundings, quantified by the upper bound $\vert R(z) \vert \lesssim (\kappa-1)$. Morally, for $\mathbb{Z}$ (or, equivalently, $\kappa = 1$), finite sections of paths are coupled to their surroundings precisely through the two vertices which form the boundary of the path section. In particular, because the boundary consists of only two vertices, taking an average over the length of the path, these contributions decay in the limit, resulting in a zero remainder term. In contrast, for the Bethe lattice ($\kappa \geq 2$), paths are coupled to their surroundings not only through the boundary vertices, but also at each interior vertex. Each interior vertex along the path is connected to $(\kappa-1)$ rooted trees and the number of these contributions grows with the length of the path. Thus, upon averaging, while this contribution is bounded, this contribution does not decay in the limit as the path length goes to infinity. Indeed, in Section \ref{sec:remainder1}, we explicitly compute the remainder term for the free Laplacian and show that its contributions are, in general, non-trivial, and thus essential.

\subsection{Detailed outline of the paper} \label{subsec:outline}

We structure the paper as follows. Section \ref{sec_set-up_det-results} serves as the backbone for the rest of the paper and has two main results. After describing some basic set-up in Section \ref{subsec:bethe1}, Section \ref{subsec:auto1} introduces a family of graph isomorphisms $\{\tau_x ~,~ x \in \mathcal{V}\}$ on the Bethe lattice which, geometrically, amount to shifts of the root to a given vertex $x$. Underlying the definition of these generalized shifts is Proposition \ref{prop_representation_vertices}, in which provide an explicit representation of the vertices as products of two elementary transformations acting on the root of the tree: a translation between coordination spheres (see (\ref{eq_transl1})) and a rotation within each coordination sphere (see (\ref{eq:transl2})). We mention that these two elementary transformations are also featured in an Appendix in \cite{AcostaKlein} however the description of the vertices in terms of products of these transformations acting on the root is incorrect (see Remark \ref{remark_acosta-klein} in Section \ref{subsec:auto1}). 
In summary, Section \ref{subsec:auto1} shows that the family of generalized shifts $\{\tau_x ~,~ x \in \mathcal{V}\}$ is finitely generated from $(\kappa+1)$ elementary shifts which are, however, non-commuting. This is different from the shifts on $\mathbb{Z}$ (or $\mathbb{Z}^d$) which leads to complications when considering limits of averages along infinite paths by trying to appeal to standard versions of the multi-parameter, non-commuting ergodic theorem such as proven by Zygmund \cite{zygmund_51}.  

Our second main result of Section \ref{sec_set-up_det-results} is given in Theorem \ref{thm_analytic-miracles} (Section \ref{subsec:rso_expansions1}): when considering the limiting behavior of averages of spectral information along infinite paths, Theorem \ref{thm_analytic-miracles} (and its extension in Corollary \ref{coro_DOSm-wek-conv}) provides a tool that allows to replace the infinite volume Schr\"odinger operator by its finite volume restrictions. In the context of ergodic Schr\"odinger operators, Theorem \ref{thm_analytic-miracles} will allow us to recover the density of states measure and the Lyapunov exponent from the respective quantities for appropriate finite volume restrictions and therefore forms key ingredients in our proof of our modified Thouless formula in Section \ref{sec:thouless1}. Theorem \ref{thm_analytic-miracles} expands on ideas by Acosta and Klein in \cite[Proposition 2.1]{AcostaKlein}, unlike these earlier results, our extension requires a suitable enlargement of the volume of the truncated lattice on which the restricted operator acts to compensate for ``surface-to-volume effects'' of the Bethe lattice. We mention that the results of Section \ref{subsec:rso_expansions1} are deterministic and are based on walk expansions of the Green function for Schr\"odinger operators on trees.

In Section \ref{sec_dos_LE1} we introduce ergodic Schr\"odinger operators on Bethe lattice and relate the ergodic structure to the Lyapunov exponent. Specifically, Section \ref{sec_erogic-structure} uses the generalized shifts introduced in Section \ref{subsec:auto1} to define ergodic Sch\"odinger operators and to describe the associated covariance relations; see (\ref{eq_unitaries})--(\ref{eq_covariance}). In Section \ref{subsec:lyapunov1}, we use the ergodic structure to encode information about infinite paths, which allows us to realize the Lyapunov exponent, defined in (\ref{eq_LE_defn_Weyl}) in terms of Weyl $m$-functions, as an exponential decay rate of the off-diagonal elements of the Green function along suitably chosen paths. On a technical level, the non-commuting nature of the underlying ergodic family, implies that we have to suitably choose paths so to enable us to use the standard multi-parameter versions of the ergodic theorem. While this limitation was necessary to guarantee existence of limits, our final results (in particular, the modified Thouless formula) do {\em{not}} depend on the paths that are chosen.

Finally, in Section \ref{sec:thouless1}, we prove (the modified version of) the Thouless formula, stated precisely in Theorem \ref{thm:thouless}. In particular, we show that the modification of the ``usual'' Thouless formula familiar from ergodic Schr\"odinger operators on $\mathbb{Z}$, which is expressed by the presence of the remainder term $R(z)$, is a result of the averaged effects of the coupling of paths to their surrounding, quantified by the upper bound $\vert R(z) \vert \lesssim (\kappa-1)$. Morally, for $\mathbb{Z}$ (or, equivalently, $\kappa = 1$), finite sections of paths are coupled to their surroundings only through the precisely two vertices which form the boundary of the path section; in particular, averaging over the length of the path, these contributions decay in the limit, resulting in a zero remainder term. In contrast, for the Bethe lattice ($\kappa \geq 2$), paths are coupled to their surrounding not merely through the boundary but also through the interior, the latter of which grows with the length of the path; thus, upon averaging, while bounded this contribution does not decay in the limit. Indeed, in Section \ref{sec:remainder1}, we explicitly compute the remainder term for the free Laplacian and show that its contributions are in general non-trivial and thus essential.


\subsection{Recent related results on the Bethe lattice}

There have been some recent advances on the structure of the phase diagram for random \Schr operators on $\mathbb{B}$.  
The structure of the phase diagram for random \Schr operators on the Bethe lattice, as described by Aizenman and Warzel \cite{aw_book,warzel_ICMP_2013} were recently confirmed by Drogin and Smart \cite{ds1} under certain conditions on the regularity of the distribution of the iid random potentials. Aggarwal and Lopatto \cite{al1} studied random \Schr operators with unbounded random potentials on the Bethe lattice with large connectivity. They proved that there are a finite number of mobility edges separating intervals of dense pure point spectrum from intervals of absolutely continuous spectrum. 

In a recent paper, Banks, Breuer, Garza-Vargas, Seelig, and Simon
\cite{BBGVSS} studied periodic Jacobi matrices on trees. As an application of their methods, they mention in \cite[Section 5]{BBGVSS} that the techniques apply to the Anderson model on regular trees They proved a Thouless-like formula \cite[ Equation \ (21)]{BBGVSS} involving the Thouless integral $\int \log |z-E| ~dn(E)$ and the Green function, but they do not relate this to the Lyapunov exponent. 





\subsection{Acknowledgement} PDH is partially supported by the Simons Foundation Collaboration Grant for Mathematicians No.\ 843327. PDH thanks the Mathematics Department of Oberlin College, and CAM thanks the Mathematics Department of the University of Kentucky, for a warm welcome during several collaborative research visits.  We thank Sam Thiel (Oberlin College) for the creating figures and Kyle Hammer (University of Kentucky) for numerical computations of the remainder term of the free Laplacian. 


\section{Schr\"odinger operators on the Bethe Lattice} \label{sec_set-up_det-results}
\setcounter{equation}{0}

In this section, we first define the Bethe lattice and describe an efficient labeling of the vertices once a root has been chosen. The homogeneity of the lattice means that any vertex can be designated as a root. We next describe the automorphism group  of $\mathbb{B}$. It is that is generated by two elements: a generalized rotation and translation. The are used to construct specific automorphisms mapping the root to any vertex.  We next review the random walk and self-avoiding walk expansions of the resolvent, and apply these in order to prove convergence of the finite-volume approximations of the Green function. Finally, these automorphisms are shown to induce unitary transformations on $\ell^2(\mathbb{B})$ and implement the covariance property of the ergodic \Schr operators. 

\subsection{Definition and Geometry of the Bethe Lattice}\label{subsec:bethe1}

Given a set $\mathcal{A}$ of vertices in $\mathbb{B}$, we let $\mathbb{B} \setminus \mathcal{A}$ denote the disconnected subgraph obtained by deleting the vertices of $\mathcal{A}$ from $\mathbb{B}$. The following two special cases will be particularly relevant to us: 
\begin{enumerate}
\item Deleting a single vertex $x \in \mathcal{V}$ from $\mathbb{B}$ results in a disjoint union of $\kappa + 1$ identical copies of rooted trees $\mathbb{B}_{y}^{(+;x)}$ whose natural root is one of the $\kappa + 1$ vertices $y$  directly connected to $x$ (denoted by $y \sim x$), 
\begin{align} \label{eq_delete-vertex}
\mathbb{B}\setminus \{x\} =: \bigsqcup_{y \sim x} \mathbb{B}_{y}^{(+;x)} ~\mbox{;}
\end{align}
we will refer to $\mathbb{B}\setminus \{x\}$ as the Bethe lattice decoupled at the vertex $\{x\}$. We recall that a rooted tree results from designating a particular vertex as the root, in which case one obtains a partial order on the tree as follows: if $y$ denotes the root, then $z,w \in \mathcal{V}$ satisfy $z < w$ if $z$ lies along the (unique) path from the root $y$ to $w$. In this sense, the $\kappa+1$ trees $\mathbb{B}_{y}^{(+;x)}$ in (\ref{eq_delete-vertex}) are naturally rooted at the respective vertex $y$ and characterized by the property that all vertices $z$ are connected to precisely $\kappa$ neighboring vertices $w$ with $z < w$ (``{\em{forward neighbors}}''). 

\item Suppose that $\gamma$ is a (finite or infinite) {\em{path}} of length $\vert \gamma \vert \in \mathbb{N} \cup \{\infty\}$. Then, for $0 \leq k \leq \vert \gamma \vert$, $\mathbb{B}\setminus \gamma([0,k])$ denotes the disconnected subgraph obtained by removing the vertices of the image $\gamma([0,k])$ from $\mathbb{B}$. For later purposes (see the RW- and SAW-expansions in Theorem \ref{prop_RW-exp} and Theorem \ref{thm_SAW-expansion}, respectively) we recall that a path of length $N =: \vert \gamma \vert \in \mathbb{N} \cup \{\infty\}$ is an {\em{injective}} map of the form $\gamma: [0, N] \cap \mathbb{N} \to \mathcal{V}$ so that for each $0 \leq k \leq N-1$, the vertices $\gamma(k)$ and $\gamma(k+1)$ are connected by an edge in $\mathcal{E}$. Paths are to be distinguished from {\em{walks}} for which injectivity of the map $\gamma$ is {\em{not}} required, so that a walk $\gamma$ may visit a vertex more than once. In view of Theorem \ref{thm_SAW-expansion}, we note that paths are also referred to as self-avoiding walks.
\end{enumerate}

Finally, we will write $\ud_\mathbb{B}$ for the natural metric on $\mathbb{B}$ which, since $\mathbb{B}$ is a tree, can be defined as follows: given two vertices $x,y \in \mathcal{V}$, $\ud_\mathbb{B}(x,y)$ is the number of edges in the {\em{unique}} path which connects $x$ to $y$. In particular, we have $x \sim y$ if and only if $\ud_\mathbb{B}(x,y) = 1$.
\begin{figure}
	\includegraphics[width = \textwidth]{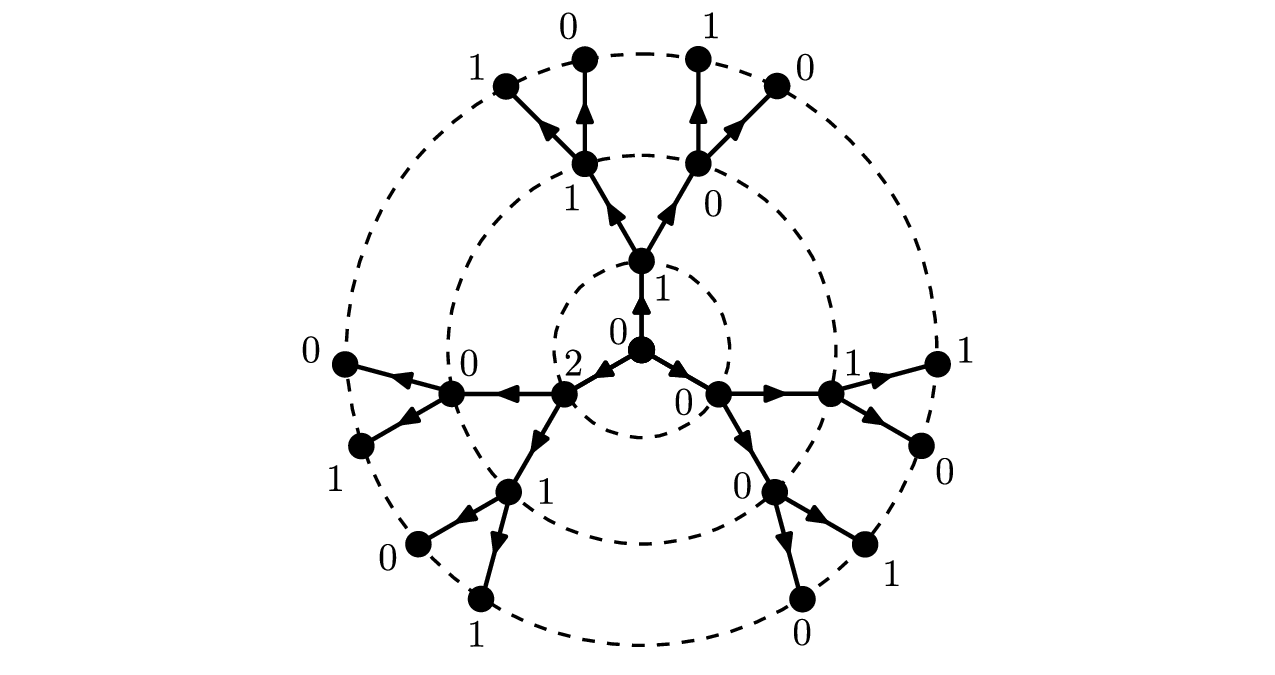} 
	\caption{Vertex labeling of the Bethe lattice for $\kappa = 2$: the information about each vertex at level $\ell \geq 1$ is uniquely encoded by the $\ell+1$-tuple $(0,a_1, \dots, a_\ell)$ where $a_1 \in \{0,1,2\}$ (corresponding to the $\kappa + 1 = 3$ forward neighbors of the root) and $a_j \in \{0,1\}$, for $2 \leq j \leq \ell$ (corresponding to the $\kappa = 2$ forward neighbors of vertices except the root). To keep the figure concise, we abbreviate this labeling by only showing the last entry $a_\ell$ of the tuple next to the respective vertex. For instance, the three vertices in the first coordination sphere $(\ell = 1)$ are labeled in the figure as $0,1,2$ to abbreviate the ordered pairs $(0,0)$, $(0,1)$, and $(0,2)$.}
	\label{fig_bethe-lattice_vertex-labeling}
\end{figure}
We introduce a labeling of the vertices of $\mathbb{B}$ as follows. We choose one vertex and label it as the origin $(0)$. We refer to the vertex $(0)$ as the root. We then obtain a rooted tree that emanates radially from the root; see Figure \ref{fig_bethe-lattice_basic}. Given the connectivity $\kappa$, $(0)$ is directly connected to $\kappa + 1$ nearest-neighbor vertices which form the first level (or coordination sphere) of $\mathbb{B}$. We label these $\kappa + 1$ vertices of the first level by an integer pair $(0,a_1)$ where $a_1 \in \{0, \dots, \kappa\}$. The second level of $\mathbb{B}$ consists of those vertices that are nearest neighbors to the vertices of the first level and have not previously been labeled. We label these vertices by a triple of integers $(0, a_1, a_2)$ with $0 \leq a_2 \leq \kappa-1$, thereby designating vertices which can be connected to $(0)$ by passing through $(0,a_1)$ for $a_1 \in \{0, \dots, \kappa\}$. Proceeding inductively, the level $\ell \geq 2$ of $\mathbb{B}$ consists of vertices labeled by an $(\ell+1)$-tuple $(0,a_1, a_2, \ldots, a_\ell)$ with $a_1 =0, \ldots, \kappa$ and $a_j = 0, \ldots, \kappa-1$, for $j=2, \ldots, \ell$. Observe that since $\mathbb{B}$ is a tree, this labeling of vertices implicitly specifies the unique path $\gamma_{x}:[0, \dots, \ell] \cap \mathbb{Z} \to \mathcal{V}$ which connects the origin $(0)$ to a given vertex $x=(0,a_1, \dots , a_\ell)$ by
\begin{equation} \label{eq_infinite-path-labeling}
\gamma_{x}(j):= (0,a_1, \dots , a_j) ~\mbox{, for $0 \leq j \leq \ell$ .}
\end{equation}
The labeling implies a radial structure of $\mathbb{B}$ for which the distance from the origin (in the metric $\ud_\mathbb{B}$) is given by the level number $\ell$; each level in this radial structure can be considered a {\em{coordination sphere}}. 
\begin{remark}
The labeling which we introduced above is similar to a labeling used by Acosta and Klein in the Appendix of \cite{AcostaKlein}. We note however the following two differences which will simplify the description of the two graph automorphisms $\tau_1, \tau_2$ defined below (see (\ref{eq_transl1}) for $\tau_1$ and (\ref{eq:transl2}) for $\tau_2$): 1. We decided to maintain the root as an explicit part of the labeling of vertices so that vertices of the $\ell$-th level are labeled by an $\ell + 1$ tuple $(0,a_1, a_2, \ldots, a_\ell)$ instead of an $\ell$ tuple as in \cite{AcostaKlein} (which omits the root $0$ in the first entry). This has the advantage that it allows to explicitly describe the transformation of the origin under the graph automorphisms $\tau_1$ and $\tau_2$ introduced below (which was omitted in \cite{AcostaKlein}). 2. We decided to let the entries assume values starting with $0$ instead of $1$. Indeed, for $\ell \in \mathbb{N}$, one may consider $a_1 \in \mathbb{Z}_{\kappa+1}$ (the class of integers modulo $\kappa+1$) and $a_j \in \mathbb{Z}_\kappa$, for $2 \leq j \leq \ell$, which simplifies the modular arithmetic in the definition of the graph automorphism $\tau_2$. 
\end{remark}

\subsection{Automorphisms of the Bethe Lattice}\label{subsec:auto1}

Following, Acosta and Klein \cite{AcostaKlein}, we use this labeling to introduce two transformations on $\mathbb{B}$. The first map $\tau_1: \mathbb{B} \to \mathbb{B}$ is a {\em{generalized level translation}}, defined by 
\begin{align} \label{eq_transl1}
 &\tau_1 (0)  =   (0,0) \nonumber \\
 & \tau_1 (0,\kappa) = (0) \nonumber \\
 & \tau_1 (0, \kappa, a_2,  \ldots, a_\ell) =   (0, (a_2+1)_{\rm{mod}(\kappa+1)}, a_3,  \ldots, a_\ell) ~\mbox{, if } \ell \geq 2  ~\mbox{, } \nonumber  \\
 & \tau_1 (0, a_1, \ldots, a_\ell)  =   (0, 0, a_1, \ldots, a_\ell) ~\mbox{, if } \ell \geq 1 ~\mbox{ and } 0  \leq  a_1 \leq \kappa - 1 ~\mbox{.}
\end{align}
Since graph automorphism preserve connected vertices, $\tau_1$ can be fully understood by its action on the root $(0)$ and the $\kappa+1$ vertices of the first level: (\ref{eq_transl1}) shifts the root and the first $\kappa$ vertices of level 1 (with $a_1 \in \{0, \dots, \kappa - 1\}$ {\em{up}} by one level; the remaining vertex in the first level with $a_1 = \kappa$ is shifted {\em{down}} to $(0)$. We illustrate the action of $\tau_1$ in Figure \ref{fig_bethe-lattice_tau1}. 
\begin{figure}
	\includegraphics[width = \textwidth]{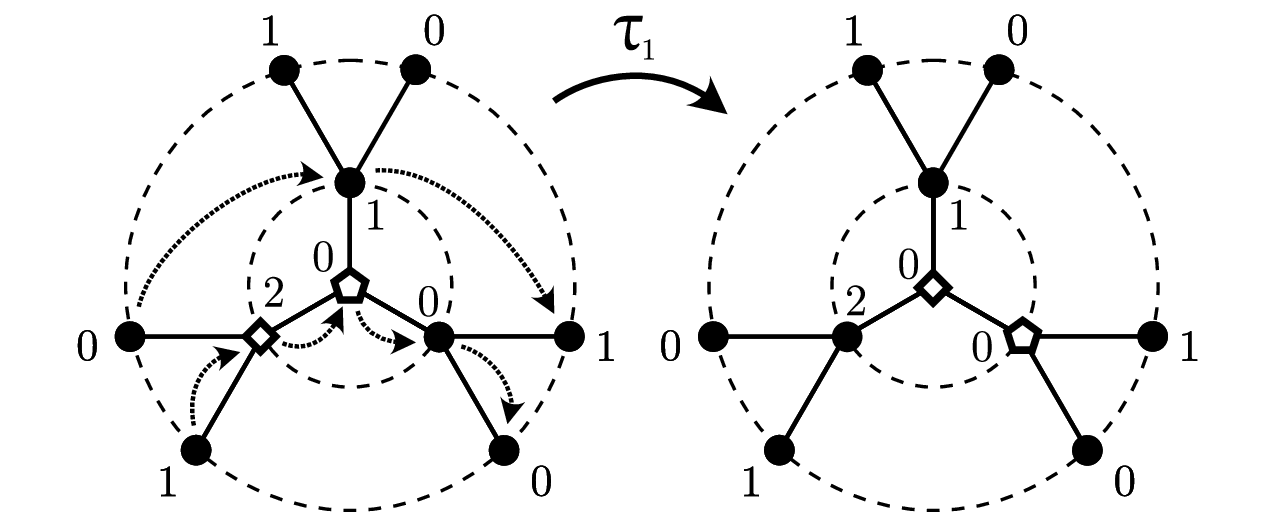} 
	\caption{Illustration of the map $\tau_1$, defined in (\ref{eq_transl1}): the graph automorphism $\tau_1$ acts as a level translation by shifting the root $(0)$ up by one one level to the vertex $(0,0)$; the remaining vertices respond accordingly (i.e., so that edges are preserved). The two panels of the figure represent a before (left panel) and after (right panel) picture of the action of $\tau_1$, where the transformation of the root $(0)$ (shown as a pentagon) and the vertex $(0,2)$ (shown as a diamond) are especially highlighted.}
	\label{fig_bethe-lattice_tau1}
\end{figure}
\begin{remark} \label{remark_tau1}
For later purposes we comment on one important detail in the definition of $\tau_1$ in (\ref{eq_transl1}), which distinguishes between vertices $x=(0, a_1, \dots , a_\ell)$ depending on whether or not $a_1$ equals $\kappa$. This distinction expresses the asymmetry between the root $(0)$, which has $\kappa+1$ forward neighbors, and the vertices at all other levels $\ell \geq 1$, which only have $\kappa$ forward neighbors. Indeed, denoting $x = (0, a_1, \dots a_\ell)$, if $a_1 \neq \kappa$ the definition in (\ref{eq_transl1}) implies that $\tau_1$ acts on $x$ as a right shift by including one additional zero in the sequence. If, however, $a_1 = \kappa$, then $\tau_1$ acts as a left shift (by deleting the entry $a_1$) and adding $1$ $\rm{mod}(\kappa +1)$ to the new first (non-trivial) entry $a_2$; see also the ``before-and-after'' picture shown in Figure \ref{fig_bethe-lattice_tau1} for the vertex $(0,2)$ displayed by a diamond symbol. The latter additional aspect of the definition ensures that $\tau_1$ is injective.
\end{remark}

The second map $\tau_2: \mathbb{B} \to \mathbb{B}$ is a {\em{generalized rotation}}, defined by
 \begin{align}\label{eq:transl2}
 & \tau_2 (0)  =   (0)  \nonumber \\
  & \tau_2 (0, a_1, \ldots, a_\ell)  =    (0, ( a_1 + 1)_{{\rm mod} (\kappa+1)}, (a_2 + 1)_{{\rm mod} (\kappa)} ,  \ldots, (a_\ell + 1)_{{\rm mod} (\kappa)}) ~\mbox{.}
 \end{align}
The action of $\tau_2$ is characterized by following two aspects: 1. $\tau_2$ preserves vertices in each level (coordination sphere); in particular it fixes the root $(0)$. 2. The action of $\tau_2$ on a given vertex $x=(0, a_1, \dots a_\ell)$ at the level $\ell \geq 1$ can be understood as $\ell$ consecutive rotations (cyclic permutations) where each vertex along the unique path (\ref{eq_infinite-path-labeling}) connecting $(0)$ to $x$ is rotated about the previous vertex on the same path, without ever breaking an edge. We illustrate the action of $\tau_2$ in Figure \ref{fig_bethe-lattice_tau2}.
\begin{figure}
	\includegraphics[width = \textwidth]{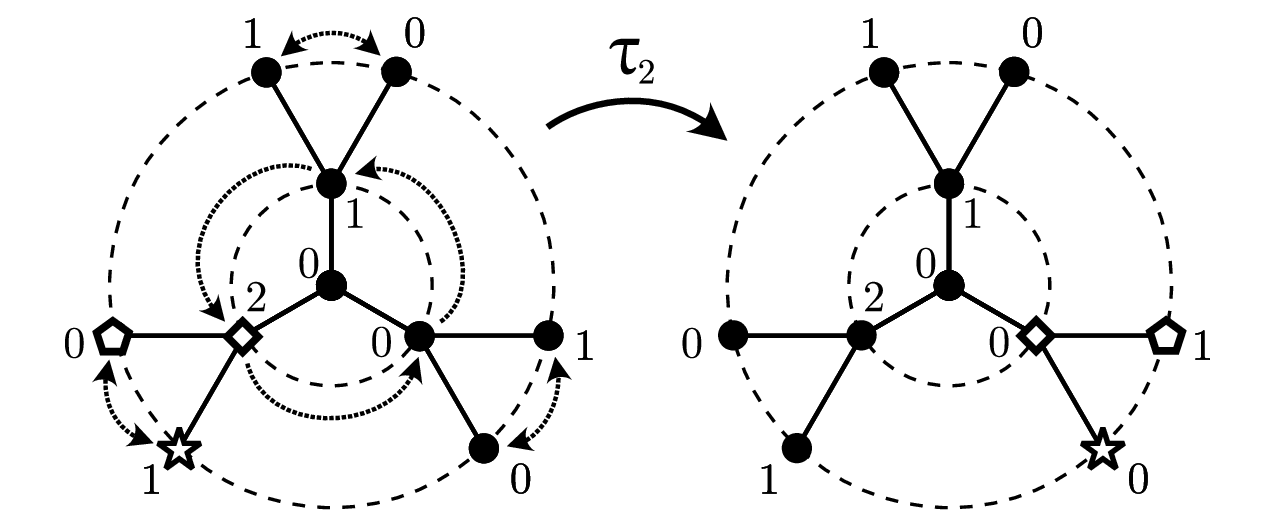} 
	\caption{Illustration of the map $\tau_2$, defined in (\ref{eq:transl2}): the action of the graph automorphism $\tau_2$ on a vertex $x$ at level $\ell \geq 1$ can be understood as the result of $\ell$ consecutive rotations (or cyclic permutations) where each vertex along the unique path connecting the root to $x$ is rotated; this rotation at each level is indicated in the left panel by dotted arrows. The side by side display of the left and right panel presents a before and after picture in which we illustrate the result of the action of $\tau_2$ for the three vertices $(0,2)$ (diamond), $(0,2,0)$ (pentagon), and $(0,2,1)$ (star).}
	\label{fig_bethe-lattice_tau2}
\end{figure}

The main characteristics of the maps $\{ \tau_1, \tau_2 \}$ are summarized as follows:
\begin{enumerate}
\item the maps $\tau_1$ and $\tau_2$ are graph automorphisms;
\item the family $\{ \tau_1, \tau_2\}$ acts transitively on $\mathbb{B}$; transitivity is implied by Proposition \ref{prop_representation_vertices} below;
\item the maps $\tau_1$ and $\tau_2$ do not commute.
\end{enumerate}
In the following proposition, we show the transitivity of the family $\{ \tau_1, \tau_2\}$ by providing an explicit representation of the vertices in terms of the action of $\{ \tau_1, \tau_2\}$ on the root; this will form a crucial ingredient for defining ergodic Schr\"odinger operators on the Bethe lattice in Section \ref{sec_dos_LE1}.
\begin{prop} \label{prop_representation_vertices}
For $\ell \in \mathbb{N}$, let $x \in \mathcal{V}$ be a vertex in the $\ell$-th level of $\mathbb{B}$. Then, $x$ can be generated by the family $\{ \tau_1, \tau_2 \}$ as the non-commutative product
\begin{equation} \label{eq_prop_representation_vertices-1}
x = (\tau_2^{d_1} \tau_1) (\tau_2^{d_{2}} \tau_1) \dots (\tau_2^{d_{\ell}} \tau_1)(0) ~\mbox{,}
\end{equation}
with exponents
\begin{equation} \label{eq_prop_representation_vertices-2}
d_1 \in \{0, \dots , \kappa\} ~\mbox{ and } d_j \in \{0, \dots , \kappa-1\} ~\mbox{, for } 2 \leq j \leq \ell ~\mbox{.}
\end{equation}
The exponents in the representation (\ref{eq_prop_representation_vertices-1})--(\ref{eq_prop_representation_vertices-2}) are {\em{unique}} and given explicitly by 
\begin{align} \label{eq_prop_representation_formulas}
d_1 = a_1 ~\mbox{; } & d_{2} = (a_2 - a_1)_{\rm{mod}(\kappa)}  ~\mbox{,} \nonumber \\
 & \vdots \nonumber \\
 & d_\ell =  (a_\ell - a_{\ell - 1})_{\rm{mod}(\kappa)} ~\mbox{.}
\end{align}
\end{prop}
\begin{remark} \label{remark_acosta-klein}
Proposition \ref{prop_representation_vertices} corrects a statement of Acosta and Klein in \cite{AcostaKlein} (top of page 303 of the Appendix) which claims that every vertex can be (uniquely) represented in the form $\tau_2^{\beta} \tau_1^{\alpha}(0)$ where $\alpha = \ell$ is the level number of $x$ and $\beta \in \mathbb{Z}$ satisfies $0 \leq \beta \leq (\kappa + 1) \kappa^{d_1 - 1}$. Indeed, using the definition of the maps in (\ref{eq_transl1})--(\ref{eq:transl2}) implies
\begin{equation}
\tau_2^{\beta} \tau_1^{\alpha}(0) = \tau_2^{\beta} (0, \underbrace{0, \dots , 0}_{\mbox{\tiny{$\alpha$ many entries}}}) = (0, (\beta)_{{\rm mod} (\kappa+1)}, \underbrace{(\beta)_{{\rm mod} (\kappa)}, \dots, (\beta)_{{\rm mod} (\kappa)}}_{\mbox{\tiny{$\alpha-1$ many entries}}} ) ~\mbox{,}
\end{equation}
which does not produce most vertices in $\mathbb{B}$. We however observe that 
one origin of the discrepancy between the representation in \cite{AcostaKlein} and (\ref{eq_prop_representation_vertices-1}) is the lack of commutativity of the maps $\tau_1$ and $\tau_2$; indeed, if $\tau_1$ and $\tau_2$ were commutative, reordering the terms in (\ref{eq_prop_representation_vertices-1}) would result in
\begin{equation}
\tau_2^{\sum_{j=1}^\ell d_j} \tau_1^\ell (0) = \tau_2^{(a_\ell - a_1)_{\rm{mod}(\kappa)} + a_1 } \tau_1^\ell (0) ~\mbox{,}
\end{equation}
which produces an expression of the form claimed in \cite{AcostaKlein}.
\end{remark}
\begin{proof}
For $\ell=1$, the claim follows by direct computation. Indeed, from the definitions of $\tau_1$ and $\tau_2$ in (\ref{eq_transl1})--(\ref{eq:transl2}) we have
\begin{equation}
\tau_2^{a_1} \tau_1(0) = \tau_2^{a_1} (0,0) = (0, (a_1)_{\rm{mod}(K)} )  = (0, a_1) = x ~\mbox{,}
\end{equation}
which is the desired unique representation.

To obtain the claim for general $\ell \geq 2$, we first establish that for all $\ell \geq 2$, it always possible to represent $x = (0 , a_1 , \dots , a_{\ell})$ in the form
\begin{equation} \label{eq_proof_prop_representation_vertices-generation_it_1}
(0, a_1, \dots a_{\ell}) = (\tau_2^{d_{1}} \tau_1) (0, b_1 , \dots b_{\ell-1}) ~\mbox{,}
\end{equation}
for a {\em{unique}} collection of integers satisfying
\begin{align} \label{eq_proof_prop_representation_vertices-generation_it_2}
d_{1} \in \{0, \dots, \kappa\} & ~\mbox{, and } b_j \in \{0, \dots, \kappa-1\} ~\mbox{, } 1 \leq j \leq \ell - 1 ~\mbox{. }
\end{align}
To do so, notice that the definitions of $\tau_1$ and $\tau_2$ in (\ref{eq_transl1})--(\ref{eq:transl2}) imply that (\ref{eq_proof_prop_representation_vertices-generation_it_1})--(\ref{eq_proof_prop_representation_vertices-generation_it_2}) hold true if and only if 
\begin{align} \label{eq_proof_prop_representation_vertices-ind}
a_1 & = d_{1} \nonumber \\
a_j & = ( b_{j-1} + d_{1} )_{\rm{mod}(\kappa)} ~\mbox{, for $2 \leq j \leq \ell$ .}
\end{align}
Distinguishing between the cases $a_1 \in \{0 , \dots, \kappa-1\}$ and $a_1 = \kappa$, (\ref{eq_proof_prop_representation_vertices-ind}) has a {\em{unique}} solution for $b_1, \dots, b_{\ell-1}$ given by
\begin{equation} \label{eq_proof_prop_representation_vertices-ind-sln}
b_j = ( a_{j+1} - a_1)_{\rm{mod}(\kappa)} ~\mbox{, for $1 \leq j \leq \ell-1$ .}
\end{equation}
The claim of Proposition \ref{prop_representation_vertices} for $\ell \geq 2$ thus follows by iteration of the representation in (\ref{eq_proof_prop_representation_vertices-generation_it_1})--(\ref{eq_proof_prop_representation_vertices-generation_it_2}) with the unique solution given by $d_1 = a_1$ and (\ref{eq_proof_prop_representation_vertices-ind-sln}).
\end{proof}

The representation of vertices as products of $\tau_1$ and $\tau_2$ given in Proposition \ref{prop_representation_vertices} allows to associate with each vertex $(0) < x \in \mathcal{V}$ a graph automorphism $\tau_x: \mathbb{B} \to \mathbb{B}$ defined by
\begin{equation} \label{eq_generalized-shift-vertex}
\tau_x:= (\tau_2^{d_1} \tau_1) (\tau_2^{d_2} \tau_1) \dots (\tau_2^{d_{\ell}} \tau_1) ~\mbox{,}
\end{equation} 
where the exponents $d_1, \dots, d_\ell$ are given in (\ref{eq_prop_representation_formulas}). For completeness, if $x = (0)$, we let $\tau_{(0)}$ denote the identity map on $\mathbb{B}$. 

Since $\tau_x(0) = x$, one can consider $\tau_x$ as a (generalized) {\em{shift to the vertex $x$}}, which moves the root $(0)$ to $x$ in a manner that preserves the overall graph structure. Indeed, for an arbitrary vertex $z$, $\tau_x(z)$ can be obtained geometrically as follows: let $\gamma_{z}$ denote the unique path connecting $(0)$ to $z$ given in (\ref{eq_infinite-path-labeling}), then $\tau_x(z)$ is the vertex obtained by appropriately attaching the path $\gamma_{z}$ to the vertex $x$ (which functions as the new root), and subsequently evaluating this path at its endpoint. We illustrate this useful geometric perspective of $\tau_x$ in Figure \ref{fig_bethe-lattice_tau_x}. 
\begin{figure}
	\includegraphics[width = \textwidth]{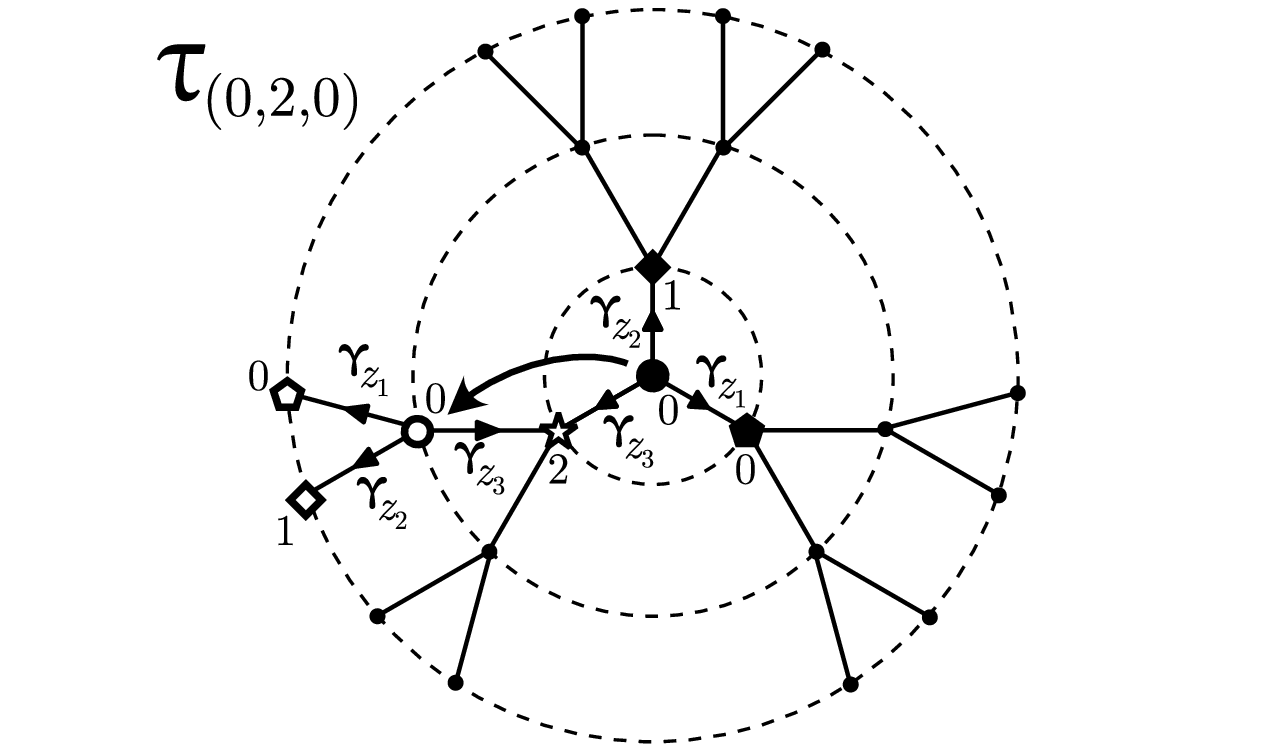} 
	\caption{Geometric perspective of the generalized shift $\tau_x$, defined in (\ref{eq_generalized-shift-vertex}): we illustrate the action of $\tau_x$ for $x = (0,2,0)$ on three vertices $z_1= (0,0)$ (solid pentagon), $z_2=(0,1)$ (solid diamond), and $z_3 = (0,2)$ (star). The results of $\tau_x(z_j)$, for $1 \leq j \leq 3$, are shown by the respective unshaded symbols. The figure illustrates that $\tau_x(z_j)$ can be obtained geometrically by appropriately attaching (as indicated by the labeling) the unique path $\gamma_{z_j}$, connecting the root to $z_j$, to the vertex $x$.  In this sense, $\tau_x$ can be viewed as shifting the root to the vertex $x$, which is shown by the arrow in the figure.}
	\label{fig_bethe-lattice_tau_x}
\end{figure}
Lemma \ref{lemma_action-maps} presents a precise formulation of this geometric perspective of the map $\tau_x$; it follows from a direct computation of $\tau_x(z)$, for an arbitrary vertex $z$, using the definitions of the maps $\tau_1$ and $\tau_2$ and the representation of vertices in Proposition \ref{prop_representation_vertices}. While the computation of $\tau_x(z)$ is straightforward, it requires distinction of various cases depending on the vertex $z$. To maintain the flow of the development, we decided to include the details to Appendix \ref{app_lemma}. Key in this computation is the telescoping property built into (\ref{eq_prop_representation_formulas})
\begin{equation} \label{eq_compositional-lemma-key}
(\sum_{j=1}^n d_j)_{\rm{mod}(\kappa)} = a_n ~\mbox{, for all $1 \leq n \leq \ell$ ,}
\end{equation}
and the distinction of the action of $\tau_1$ depending on whether or not the first (non-trivial) entry equals $\kappa$ (see Remark \ref{remark_tau1} below). In the statement of Lemma \ref{lemma_action-maps}, we use the convention that sums of the form $\sum_{j = \alpha}^{\beta} x_j$ with $\alpha > \beta$ evaluate to zero.
\begin{lemma} \label{lemma_action-maps}
Consider a fixed vertex $x = (0, a_1, \dots, a_\ell)$ with $\ell \geq 1$ and let $d_1, \dots d_\ell$ be defined as in (\ref{eq_prop_representation_formulas}) of Proposition \ref{prop_representation_vertices}. Suppose that $z = (0, b_1, \dots b_m)$ is an arbitrary vertex where $m \geq 1$. Then, $\tau_x(z)$ is described by the following:
\begin{itemize}
\item[(i)] If $b_1 \neq \kappa$, then
\begin{equation}
\tau_x(z) = (0, a_1, \dots , a_\ell, (b_1 + a_\ell)_{\rm{mod}(\kappa)}, (b_2 + a_\ell)_{\rm{mod}(\kappa)}, \dots (b_m + a_\ell)_{\rm{mod}(\kappa)}) ~\mbox{.}
\end{equation}
\item[(ii)] If $b_1 = \kappa$, then for $m=1$, one has
\begin{equation}
\tau_x(0, b_1) = (0, a_1, \dots , a_{\ell-1}) ~\mbox{.}
\end{equation}
\item[(iii)] For $m \geq 2$, suppose that $b_1 = \kappa$ and that $n$ is the {\em{smallest}} index with $1 \leq n \leq \min\{m , \ell\} - 1$ which satisfies the condition that
\begin{align}
c_{n+1} := \left[ \left(b_{n+1} + \sum_{j=\ell - (n-2)}^\ell d_j \right)_{\rm{mod}(\kappa)} + 1 + d_{\ell - (n-1)} \right]_{\rm{mod} (\kappa + 1)} \neq \kappa ~\mbox{, }
\end{align}
then, one has
\begin{align}
\tau_x(z)=(0, a_1, \dots, a_{\ell - n} , (c_{n+1} + a_{\ell - n})_{\rm{mod}(\kappa)}, (b_{n+2} + a_\ell)_{\rm{mod}(\kappa)}, \dots (b_m + a_\ell)_{\rm{mod}(\kappa)}) ~\mbox{.}
\end{align}
\item[(iv)] For $m \geq 2$, suppose that $b_1 = \kappa$ and for all indices $n$ with $1 \leq n \leq \min\{m , \ell\}$, one has
\begin{align}
\left[ \left(b_{n+1} + \sum_{j=\ell - (n-2)}^\ell d_j \right)_{\rm{mod}(\kappa)} + 1 + d_{\ell - (n-1)} \right]_{\rm{mod} (\kappa + 1)} = \kappa ~\mbox{.}
\end{align}
Then, if $m \leq \ell$, one has
\begin{equation}
\tau_x(z)= (0, a_1, \dots , a_{\ell-m}) ~\mbox{,}
\end{equation}
and, if $m > \ell$, one has
\begin{equation}
\tau_x(z)= (0, \left[(b_{\ell+1} +a_\ell - a_1)_{\rm{mod} (\kappa)} + 1 + a_1\right]_{\rm{mod}(\kappa+1)}, (b_{\ell+1} + a_\ell)_{\rm{mod} (\kappa)}, \dots, (b_{m} + a_\ell)_{\rm{mod} (\kappa)}) ~\mbox{.}
\end{equation}
\end{itemize}
\end{lemma}
We remark that item (iii) of Lemma \ref{lemma_action-maps} expresses the scenario that the first $n$ entries of $z$ result in $\tau_1$ acting as a left shift, after which, since $c_{n+1} \neq \kappa$, $\tau_1$ acts as a right shift.

The geometric perspective of the action of $\tau_x$ as a generalized shift to the vertex $x$ (illustrated in Figure \ref{fig_bethe-lattice_tau_x}) immediately implies that following compositional property which relates the elements of the family $\{\tau_x ~,~ x \in \mathcal{V} \}$:
\begin{equation} \label{eq_compositional-law}
\tau_x \circ \tau_y = \tau_{\tau_x(y)} ~\mbox{, for all $x , y \in \mathcal{V}$ .}
\end{equation}
In view of our later discussion of ergodic Schr\"odinger operators on $\mathbb{B}$ (Section \ref{sec_dos_LE1}), we note that given (\ref{eq_compositional-law}), one also has 
\begin{equation} \label{eq_compositional-law-1}
\tau_x^{-1} \circ \tau_y = \tau_{\tau_x^{-1}(y)} ~\mbox{, for all $x , y \in \mathcal{V}$ .}
\end{equation}
Indeed, to see that (\ref{eq_compositional-law-1}) is implied by (\ref{eq_compositional-law}), observe that (\ref{eq_compositional-law-1}) is equivalent to
\begin{equation}
\tau_y = \tau_x \circ \tau_{\tau_x^{-1}(y)} ~\mbox{, for all $x , y \in \mathcal{V}$ ,}
\end{equation}
which however follows from (\ref{eq_compositional-law}) because
\begin{equation}
\tau_x \circ \tau_{\tau_x^{-1}(y)} = \tau_{\tau_x(\tau_x^{-1}(y))} = \tau_y ~\mbox{.}
\end{equation}

\subsection{Expansions of the Green Function: Limiting Behavior of Finite Volume Restrictions}\label{subsec:rso_expansions1}

Based on the set-up from Sections \ref{subsec:bethe1}--\ref{subsec:auto1}, we first introduce deterministic Schr\"odinger operators on the Bethe lattice: given a real-valued sequence $V \in \ell^\infty(\mathbb{B}; \mathbb{R})$ (``\emph{the potential sequence}''), we let $H:~\ell^2(\mathbb{B}; \C) \to \ell^2(\mathbb{B}; \C)$ be the bounded self-adjoint operator defined by
\begin{equation} \label{eq_Schrodop}
H:= \Delta_\mathbb{B} + T_V ~\mbox{,}
\end{equation}
where, $T_V: \ell^2(\mathbb{B}; \C) \to \ell^2(\mathbb{B}; \C)$ denotes the multiplication operator by the sequence $V$,
\begin{equation}
(T_V \psi)(x) = V(x) \psi(x), ~~ \forall x \in \mathcal{V} ~\mbox{,}
\end{equation}
and $\Delta_\mathbb{B}: \ell^2(\mathbb{B}; \C) \to \ell^2(\mathbb{B}; \C)$ is the discrete Laplacian on $\mathbb{B}$,
\begin{equation} \label{eq_laplacian}
[\Delta_\mathbb{B}\psi](x):= \sum_{y \sim x} \psi(y),  ~~ \forall x \in \mathcal{V}  ~\mbox{.}
\end{equation}
It will be also be useful to consider restrictions of (\ref{eq_Schrodop}) to subgraphs: given a set of vertices $\mathcal{A}$, we let  $H_{\mathcal{A}}$ denote the restriction of $H$ to $\ell^2(\mathcal{A}; \mathbb{C})$ with Dirichlet boundary conditions, 
\begin{equation} \label{eq_defn_restriction}
H_{\mathcal{A}} := P_{\mathcal{A}} ~H ~P_{\mathcal{A}} ~\mbox{,}
\end{equation}
where $P_{\mathcal{A}}$ is the orthogonal projection onto the subspace $\ell^2(\mathcal{A}; \mathbb{C})$. One special case of (\ref{eq_defn_restriction}) are restrictions to subgraphs obtained by deleting a set of vertices from $\mathbb{B}$, as defined earlier. In particular, for a fixed vertex $x \in \mathcal{V}$, the decomposition in (\ref{eq_delete-vertex}) implies that
\begin{equation} \label{eq_rooted_hamiltonian_zero}
H= \left\{ V(x) \vert \delta_x \rangle \langle \delta_x \vert \ + \sum_{y \sim x} \vert \delta_y \rangle \langle \delta_y \vert \right\} + \sum_{y \sim x} H_y^{+; x}  ~\mbox{,}
\end{equation}
where $\{\delta_y ~,~ y \in \mathcal{V}\}$ is the usual standard basis on $\ell^2(\mathbb{B}; \mathbb{C})$. Finally, given the definition of $H_{\mathcal{A}}$ in (\ref{eq_defn_restriction}), for $z \in \mathbb{C} \setminus \mathbb{R}$, we define the associated resolvent operator by
\begin{equation} \label{eq_restrictions-greensf}
G_{\mathcal{A}}(z):= \left[ (H_{\mathcal{A}} - z)\vert_{\ell^2(\mathcal{A}; \mathbb{C})} \right]^{-1} \oplus 0_{\ell^2(\mathbb{B} \setminus \mathcal{A}; \mathbb{C})} ~\mbox{,}
\end{equation}
where $0_{\ell^2(\mathbb{B} \setminus \mathcal{A}; \mathbb{C})}$ is the zero operator on $\ell^2(\mathbb{B} \setminus \mathcal{A}; \mathbb{C})$. 

The main results of this paper are based on two known expansions of the resolvent operator (Theorem \ref{eq_prop_RW-exp} and Theorem \ref{thm_SAW-expansion}), which hold true more generally for Schr\"odinger operators on graphs. We will use these expansions to relate the resolvent operators of the ``infinite volume'' operator $H$ to the resolvent operators of ``finite volume restrictions,'' 
\begin{equation} \label{eq_finite-volume-restriction}
G_L(z):=\left[ H_{\Lambda_L} - z \right]^{-1} ~\mbox{ ,}
\end{equation}
where, for $L \in \mathbb{N}$, we let
\begin{equation} \label{eq_coordination-spheres_defn}
\Lambda_L:= \left\{ y \in \mathcal \mathcal{V} ~:~\ud_\mathbb{B}(y , 0) \leq L \right\} ~\mbox{.}
\end{equation}

The first expansion of the resolvent is known as the {\em{random walk (RW)-expansion}}, see, e.g., \cite[Theorem 6.1]{aw_book}. Since we will later use elements of the proof of the RW-expansion, we include a short argument below. For this purpose, we let $\mathbb{G} = (\mathcal{V}_\mathbb{G}, \mathcal{E}_\mathbb{G})$ denote a general graph with vertices $\mathcal{V}_\mathbb{G}$ and edges $\mathcal{E}_\mathbb{G}$ for which one can consider deterministic Schr\"odinger operators $H$ defined in complete analogy to (\ref{eq_Schrodop}); in particular, $\Delta_\mathbb{G}$ is the Laplacian on $\mathbb{G}$ given by
\begin{equation} \label{eq_laplacian_general-graph}
[\Delta_\mathbb{G}\psi](x):= \sum_{y \sim x} \psi(y),  ~\forall x \in \mathcal{V}  ~\mbox{.}
\end{equation}
It is enough for our purposes to consider the case that $\mathbb{G}$ has bounded degree, i.e., there exists $N \in \mathbb{N}$ such that the number of edges emanating from each vertex is uniformly bounded above by $N$. In this case, $\Delta_\mathbb{G}$ is a bounded self-adjoint operator on $\ell^2(\mathbb{G}; \mathbb{C})$ satisfying the elementary bound
\begin{equation} \label{eq_laplacian_normbound_trivial}
\Vert \Delta_\mathbb{G} \Vert \leq N ~\mbox{.}
\end{equation}
For the special case of the Bethe lattice with connectivity $\kappa$ and finite volume restrictions of the Laplacian $\Delta$ to balls $\Lambda_L$, the bound in (\ref{eq_laplacian_normbound_trivial}) implies that
\begin{equation} \label{eq_trivial-norm-bounds-Laplacian}
\Vert \Delta_{\Lambda_L} \Vert \leq \kappa+1 ~\mbox{, for all $L \in \mathbb{N} \cup \{\infty\}$ , }
\end{equation}
where, for ease of notation, we take 
\begin{equation} \label{eq_convention_infinitebox}
\Lambda_L = \mathbb{B} ~\mbox{ if }  L=\infty ~\mbox{.}
\end{equation}
\begin{theorem}[RW-expansion of the resolvent]
\label{prop_RW-exp}
Let $z \in \mathbb{C} \setminus \mathbb{R}$ with $\vert \Im z \vert > \Vert \Delta_\mathbb{G} \Vert$. Then, for each pair of vertices $x, y \in \mathcal{V}_\mathbb{G}$, one has the absolutely convergent expansion
\begin{equation} \label{eq_prop_RW-exp}
G(z; x,y) = \langle \delta_x , (H - z)^{-1} \delta_y \rangle = \sum_{n=0}^\infty \sum_{\gamma: x \to y, \vert \gamma \vert = n} (-1)^n \prod_{k=0}^{n} \dfrac{1}{V(\gamma(k)) - z} ~\mbox{,}
\end{equation}
where $\gamma: x \to y$,  $\vert \gamma \vert = n$ denotes all walks from the vertex $x$ to the vertex $y$ with length $\vert \gamma \vert$ fixed to $n \in \mathbb{N} \cup \{0\}$.
\end{theorem}
The proof of Theorem \ref{prop_RW-exp} results from expanding $(H - z)^{-1}$ with respect to the Laplacian where the restriction $\vert \Im z \vert > \Vert \Delta_\mathbb{G} \Vert$ ensures convergence of this expansion in operator norm; see (\ref{eq_exp_RW-pert_norm-conv}) in the proof below.
\begin{proof}
Using the second resolvent identity repeatedly with $\Delta_\mathbb{G}$ viewed as a perturbation, we obtain for each $n \in \mathbb{N}$,
\begin{align} 
(H & - z)^{-1} = (V - z)^{-1} - (V - z)^{-1} \Delta_\mathbb{G} (H - z)^{-1} \\
 &= \sum_{k=0}^{n-1} \left[ - (V - z)^{-1} \Delta_\mathbb{G}  \right]^k  (V - z)^{-1} + R_n ~\mbox{,} \label{eq_exp_RW-pert}
\end{align}
where the remainder term $R_n$ is given by
\begin{equation} \label{eq_exp_RW-pert_remainder}
R_n:= \left[- (V - z)^{-1} \Delta_\mathbb{G}  \right]^{n} (H - z)^{-1} ~\mbox{.}
\end{equation}
Since, for $\vert \Im z \vert > \Vert \Delta_\mathbb{G} \Vert$, the remainder satisfies the bound,
\begin{equation}
\left\Vert R_n \right\Vert \leq \frac{1}{\vert \Im z \vert} \left(  \dfrac{\Vert \Delta_\mathbb{G} \Vert}{ \vert \Im z \vert } \right)^{n} \to 0 ~\mbox{, as $n \to \infty$ ,}
\end{equation} 
(\ref{eq_exp_RW-pert}) yields a norm convergent expansion of the resolvent operator
\begin{equation} \label{eq_exp_RW-pert_norm-conv}
(H - z)^{-1} = \sum_{n=0}^{\infty} \left[ - (V - z)^{-1} \Delta_\mathbb{G}  \right]^n (V - z)^{-1} ~\mbox{.}
\end{equation}
In particular, computing matrix elements $G(z; x,y)$ based on (\ref{eq_exp_RW-pert_norm-conv}), the nearest neighbor of the Laplacian in (\ref{eq_laplacian_general-graph}) results in the walk-expansion claimed by (\ref{eq_prop_RW-exp}).
\end{proof}

We note that the RW-expansion of Proposition \ref{prop_RW-exp} has contributions from {\em{all}} walks which connect the vertices $x$ and $y$ under consideration. Here, we recall that walks may visit vertices more than once. For vertices $x \neq y$, it is, however, possible to carry out a partial re-summation of the terms in the RW-expansion, thereby reducing the expansion to only include self-avoiding walks (= paths). The resulting expansion is known as the self-avoiding walk expansion (SAW-expansion) or Feenberg expansion, see, e.\ g. \cite[Theorem 6.2]{aw_book}. 
We observe that a similar result was obtained independently by Bourgain and Jitomirskaya in their work on the weak-coupling regime of quasi-periodic Schr\"odinger operators; see 
\cite[Lemma 19]{BourgainJitomirskaya_2002_inventiones}. 

\begin{theorem}[SAW-expansion of the resolvent]
\label{thm_SAW-expansion}
Let $z \in \mathbb{C} \setminus \mathbb{R}$ with $\vert \Im z \vert > \Vert \Delta_\mathbb{G} \Vert$. Then, for each pair of vertices $x \neq y \in \mathcal{V}_\mathbb{G}$, one has the absolutely convergent expansion
\begin{equation} \label{eq_thm_SAW-expansion}
G(z; x,y) = \sum_{\gamma: x \to y}^{\mbox{\tiny{(SAW)}}} (-1)^{\vert \gamma \vert} \prod_{k=0}^{\vert \gamma \vert} \left\langle \delta_{\gamma(k)}, G_{\mathbb{G} \setminus \gamma([0,k-1])}(z)  \delta_{\gamma(k)} \right\rangle ~\mbox{.}
\end{equation}
Here, the (possibly infinite) sum in (\ref{eq_thm_SAW-expansion}) ranges over all {\em{paths}} (or self-avoiding walks) $\gamma: x \to y$ from $x$ to $y$, which we indicate explicitly by the superscript (SAW) in the summation. Given such a path $\gamma$ of length $\vert \gamma \vert$, for each $1 \leq k \leq \vert \gamma \vert$, $G_{\mathbb{G} \setminus \gamma([0,k-1])}(z)$ denotes the resolvent operator associated with the restriction of $H$ to $\ell^2(\mathbb{G} \setminus \gamma([0, k-1])$ as defined earlier in (\ref{eq_restrictions-greensf}); for $k=0$, we adapt the notation $G_{\mathbb{G} \setminus \gamma([0,k-1])}(z) := G_\mathbb{G}(z)$.
\end{theorem}
\begin{remark} \label{remark_SAW-exp}
Note that the SAW-expansion takes a particularly simple form if $\mathbb{G}$ is a tree in which case each pair of {\em{distinct}} vertices $x \neq y$ are connected by a {\em{unique}} path so that the sum on the right-hand side of (\ref{eq_thm_SAW-expansion}) collapses to a {\em{single}} term. Given that both sides of the SAW-expansion (\ref{eq_thm_SAW-expansion}) are analytic on $\mathbb{C} \setminus \mathbb{R}$, the uniqueness theorem for analytic functions correspondingly implies that (\ref{eq_thm_SAW-expansion}) extends from $\vert \Im (z) \vert > \Vert G \Vert$ to all of $\mathbb{C} \setminus \mathbb{R}$. This observation will be crucial in our proof of Theorem \ref{thm_analytic-miracles} (ii) below, which constitutes one of the key ingredients in this paper. 
\end{remark}

Applied to the Bethe lattice, the RW- and the SAW-expansion allow us to prove the following two important deterministic results, formulated in Theorem \ref{thm_analytic-miracles}: they compare the resolvent operator of finite volume restrictions of the Schr\"odinger operator (\ref{eq_finite-volume-restriction}) to the resolvent of the infinite volume operator $G(z)$. In the context of ergodic Schr\"odinger operators (which will be discussed in Section \ref{sec_dos_LE1}), Theorem \ref{thm_analytic-miracles} will allow us to recover the density of states measure (part (i) of Theorem \ref{thm_analytic-miracles}) and the Lyapunov exponent (part (ii) of Theorem \ref{thm_analytic-miracles}) from the respective quantities for appropriate finite volume restrictions. They will form key ingredients in our proof of a modified Thouless formula in Section \ref{sec:thouless1}. 

Theorem \ref{thm_analytic-miracles} extends a result by Acosta and Klein in \cite[Proposition 1.2]{AcostaKlein} which establishes the convergence, in the strong operator topology, of the resolvent of finite volume restriction of the \Schr operator 
to the resolvent of the infinite volume operator. We generalize this result for finite volume restrictions of the resolvent operator to characterize the behavior of 1.) averaged traces over a (growing) set of vertices $\mathcal{A}_L \subseteq \Lambda_L$, as $L \to \infty$ (part (i) of Theorem \ref{thm_analytic-miracles}), and 2.) the exponential decay rate of the Green function along (infinite) paths (part (ii) of Theorem \ref{thm_analytic-miracles}). 
Unlike earlier results by Acosta and Klein in \cite{AcostaKlein}, our extensions will require enlarging the volume of the restricted operator by a growing amount expressed by a function $\mathfrak{e}: \mathbb{N} \to \mathbb{N}$, $\mathfrak{e}(L) \to +\infty$, as $L \to \infty$. The necessary enlargement when passing from quantities relating to the infinite volume operator to finite volume restrictions is an expression of the hyperbolic geometry of the Bethe lattice which, for instance, is encapsulated by the {\em{bounded}} (non-decaying) surface to volume ratio of balls $\Lambda_L$ in the limit as $L \to +\infty$.
\begin{theorem} \label{thm_analytic-miracles}
Fix $\mathfrak{e}: \mathbb{N} \to \mathbb{N}$ such that $\mathfrak{e}(L) \to +\infty$, as $L \to \infty$.
\begin{itemize}
\item[(i)] For $L \in \mathbb{N}$, let $\mathcal{A}_L$ be a set of vertices satisfying 
\begin{equation}
\mathcal{A}_L \subseteq \Lambda_L ~\mbox{,}
\end{equation}
and denote by $\vert \mathcal{A}_L \vert$ the number of vertices in $\mathcal{A}_L$. Then, for all $z \in \mathbb{C}$ with $\Im z > 0$, one has:
\begin{align} \label{eq_thm_analytic-miracles_DOS}
\lim_{L \to \infty} \frac{1}{\vert \mathcal{A}_L \vert} \sum_{x \in \mathcal{A}_L} \left\{ G(z; x,x) - G_{L + \mathfrak{e}(L)}(z;x,x) \right\}     = 0 ~\mbox{.}
\end{align}

\item[(ii)] Let $\gamma: \mathbb{N}_0 \to \mathcal{V}$ be an infinite path with $\gamma(0) = 0$. Then, for all $z \in \mathbb{C}$ with $\Im z > 0$, one has:
\begin{align} \label{eq_thm_analytic-miracles_LE}
\lim_{k \to \infty} \frac{1}{L} \left\{ \log \vert G\left(z; 0, \gamma(L-1)\right) \vert - \log \vert G_{L + \mathfrak{e}(L)}\left(z; 0, \gamma(L-1)\right) \vert \right\} = 0 ~\mbox{.}
\end{align}
\end{itemize}
\end{theorem}
In view of part (ii) of Theorem \ref{thm_analytic-miracles}, we note that since, by definition, paths must not visit a vertex more than once, the tree structure of the Bethe lattice implies that any infinite path $\gamma$ in $\mathbb{B}$ with $\gamma(0) = 0$ necessarily satisfies 
\begin{equation}
\ud_\mathbb{B}(0, \gamma(k)) = k ~\mbox{, for all $k \in \mathbb{N}_0$ ,}
\end{equation}
so that
\begin{equation}
\ud_\mathbb{B}(0, \gamma(k)) \nearrow + \infty ~\mbox{, as $k \to \infty$ .}
\end{equation}

The proof of Theorem \ref{thm_analytic-miracles} will establish the claims first on the half plane $\Im z >  \kappa +1$ for which, by (\ref{eq_trivial-norm-bounds-Laplacian}), the RW-expansion converges absolutely. Analyticity of the resolvent will then allow us to extend the result to the entire upper half plane $\Im z > 0$ by appealing to Vitali's convergence theorem from complex analysis (for reference, see e.g., \cite[Theorem 6.2.8]{simon_analysis-books_basic-complex-analysis}). This argument further develops ideas which played an important role in \cite[Proposition 1.2]{AcostaKlein}.
\begin{proof}
Fix an arbitrary $z \in \mathbb{C}$ with $\Im z > \kappa + 1$ so that the RW-expansion for $G_{L + \mathfrak{e}(L)}(z;x,y)$ converges absolutely for all $L \in \mathbb{N} \cup \{\infty\}$; for $L=\infty$, we recall our earlier convention introduced in (\ref{eq_convention_infinitebox}), in which case $G_\infty(z; x,x) = G(z; x,x)$.

\begin{itemize}
\item[(i)] To begin, observe that for each vertex $x \in \mathcal{A}_L$, the RW-expansion implies that the only contributions to the difference between $G(z;x,x)$ and $G_{L+\mathfrak{e}(L)}(z;x,x)$ stem from self-connecting walks $\gamma: x \to x$ which {\em{escape}} $\Lambda_{L+\mathfrak{e}(L)}$; in particular the length of the contributing walks is constrained by
\begin{equation}
\vert \gamma \vert \geq 2 (\mathfrak{e}(L) + 1) ~\mbox{.}
\end{equation}
Thus, for all $x \in \mathcal{A}_L$, the norm convergence in (\ref{eq_exp_RW-pert_norm-conv}) which underlies the RW-expansion implies that
\begin{align}
\left\vert G(z;\right. & \left. x,x)  - G_{L + \mathfrak{e}(L)}(z; x,x) \right\vert = \left\vert \sum_{n = 2 (\mathfrak{e}(L) + 1) }^\infty \left\langle \delta_x ,\left[ - (V - z)^{-1} \Delta_\mathbb{B} \right]^n  (V - z)^{-1}  \delta_x \right\rangle \right. \nonumber  \\
& \left. - \sum_{n = 2 (\mathfrak{e}(L)+1)}^\infty \left\langle \delta_x , \left[ - (V - z)^{-1} \Delta_{\Lambda_{L + \mathfrak{e}(L)}}  \right]^n (V - z)^{-1}  \delta_x \right\rangle \right\vert \\
& \leq \frac{2}{\vert \Im z \vert} \sum_{n = 2 (\mathfrak{e}(L) + 1)}^\infty \left(  \dfrac{K}{ \vert \Im z \vert } \right)^{n} \to 0 ~\mbox{, as $L \to +\infty$ ;} \label{eq_thm_analytic-miracles_keybound}
\end{align}
we note that it is in (\ref{eq_thm_analytic-miracles_keybound}) that we used our hypothesis that $\mathfrak{e}(L) \to \infty$, as $L \to \infty$. Since the right-most bound in (\ref{eq_thm_analytic-miracles_keybound}) holds uniformly for all $x \in \mathcal{A}_L$, we obtain the claim in (\ref{eq_thm_analytic-miracles_DOS}) for $\Im z > \kappa + 1$. 

Finally, analyticity of the resolvent on $\mathbb{C} \setminus \mathbb{R}$ and the trivial upper bound
\begin{equation}
\vert G_{L + \mathfrak{e}(L)}(z; x,x) \vert \leq \dfrac{1}{\vert \Im z \vert} ~\mbox{, for all $L \in \mathbb{N} \cup \{\infty\}$ ,}
\end{equation}
shows that
\begin{equation}
z \mapsto \frac{1}{\vert \mathcal{A}_L \vert} \sum_{x \in \mathcal{A}_L} \left\{ G(z; x,x) - G_{2L}(z;x,x) \right\} ~\mbox{,}
\end{equation}
is a normal family of analytic functions on $\{ \Im z > 0 \}$, indexed by $L \in \mathbb{N}$. Thus, using Vitali's convergence theorem from complex analysis, (\ref{eq_thm_analytic-miracles_DOS}) extends from $\Im z > \kappa+1$ to all of $\{ \Im z > 0 \}$.

\item[(ii)] Let $\gamma: \mathbb{N}_0 \to \mathcal{V}$ be a fixed infinite path with $\gamma(0) = 0$. We will use the SAW-expansion to express the off-diagonal elements of the two resolvent operators in (\ref{eq_thm_analytic-miracles_LE}). Recall from Remark \ref{remark_SAW-exp} that the tree structure of the Bethe lattice implies that the SAW-expansion for off-diagonal elements of the resolvent extends to {\em{all}} $z \in \mathbb{C} \setminus \mathbb{R}$. Thus, for all $z \in \mathbb{C}$ with $\Im z > 0$, we obtain
\begin{align}
\dfrac{1}{L} \left\{ \right. & \left. \log    \left\vert G \left(z; 0, \gamma(L-1)\right) \right\vert -  \log \left\vert G_{L+\mathfrak{e}(L)}\left(z; 0, \gamma(L-1)\right) \right\vert \right\} \nonumber \\
& = \dfrac{1}{L} \sum_{k=0}^{L-1} \left\{ \log \left\vert G_{ \mathbb{B} \setminus \gamma([0, k-1])   }(z; \gamma(k), \gamma(k)) \right\vert  - \log \left\vert G_{ L+\mathfrak{e}(L); \mathbb{B} \setminus \gamma([0, k-1])   }(z; \gamma(k), \gamma(k)) \right\vert   \right\} ~\mbox{.} \label{eq_proof_thm_analytic-miracles_LE-start}
\end{align}

Further, recall that for every bounded self-adjoint operator $H$ on a Hilbert space $\mathscr{H}$, the resolvent is an operator-valued Herglotz function so that for all $\psi \in \mathscr{H}$ and $z = E+i\epsilon$ with $\epsilon > 0$, one has
\begin{align} \label{eq_proof_thm_analytic-miracles_bounds-greens-diag}
\dfrac{\epsilon}{\epsilon^2 + (\vert E \vert + \Vert H \Vert)^2 } & \leq \Im \left\langle\psi , (H- (E+i\epsilon))^{-1} \psi \right\rangle \nonumber \\
& \leq \left\vert \left\langle\psi , (H- (E+i\epsilon))^{-1} \psi \right\rangle \right\vert  \leq \dfrac{1}{\epsilon} ~\mbox{.}
\end{align}
In particular, letting $z \mapsto \log(z)$ denote the main branch of the complex logarithm with domain of analyticity given by $\mathbb{C} \setminus(-\infty, 0]$, we can view the right-hand side of (\ref{eq_proof_thm_analytic-miracles_LE-start}) as the real part of the function
\begin{equation}
g_L(z):= \dfrac{1}{L} \sum_{k=0}^{L-1} \left\{ \log G_{ \mathbb{B} \setminus \gamma([0, k-1])   }(z; \gamma(k), \gamma(k))   - \log G_{ L+\mathfrak{e}(L); \mathbb{B} \setminus \gamma([0, k-1])   }(z; \gamma(k), \gamma(k))   \right\} ~\mbox{,}
\end{equation}
which is analytic for all $\Im(z) > 0$. Following, we will prove that for all $z$ with $\Im(z) > 0$, one has
\begin{equation} \label{eq_proof_thm_analytic-miracles_LE-reduction}
g_L(z) \to 0 ~\mbox{, $L \to \infty$ ,}
\end{equation}
in which case, 
\begin{equation}
\lim_{L \to \infty} \dfrac{1}{L} \sum_{k=0}^{L-1} \left\{ \log \left\vert G_{ \mathbb{B} \setminus \gamma([0, k-1])   }(z; \gamma(k), \gamma(k)) \right\vert  - \log \left\vert G_{ L+\mathfrak{e}(L); \mathbb{B} \setminus \gamma([0, k-1])   }(z; \gamma(k), \gamma(k)) \right\vert   \right\} = 0~\mbox{, }
\end{equation}
for all $z \in \mathbb{C}$ with for all $\Im(z) > 0$; by (\ref{eq_proof_thm_analytic-miracles_LE-start}) this will establish our claim in (\ref{eq_thm_analytic-miracles_LE}).

To prove (\ref{eq_proof_thm_analytic-miracles_LE-reduction}), first observe that for each $0< \delta_1 < \delta_2$, the bounds in (\ref{eq_proof_thm_analytic-miracles_bounds-greens-diag}) and (\ref{eq_trivial-norm-bounds-Laplacian}) show that for each compact subset $A \subseteq \{\delta_1 \leq \Im(z) \leq \delta_2\}$, there exist constants $0 < C_1 < C_2 < +\infty$ such that for all $L \in \mathbb{N}$ and all $0 \leq k \leq L-1$, one has 
\begin{align}
C_1 \leq \Im G_{ L+\mathfrak{e}(L); \mathbb{B} \setminus \gamma([0, k-1]}(z; \gamma(k), \gamma(k)) \leq  \vert G_{ L+\mathfrak{e}(L); \mathbb{B} \setminus \gamma([0, k-1]}(z; \gamma(k), \gamma(k)) \vert \leq C_2 ~\mbox{, } \nonumber \\
C_1 \leq \Im G_{\mathbb{B} \setminus \gamma([0, k-1]}(z; \gamma(k), \gamma(k)) \leq \vert G_{\mathbb{B} \setminus \gamma([0, k-1]}(z; \gamma(k), \gamma(k)) \vert \leq C_2 ~\mbox{, } \label{eq_bounds_lower-upper-resolvent-diagonal}
\end{align}
for all $z \in A$. Since by analyticity, $\log(z)$ is Lipschitz on compact subsets of $\{\Im(z) >0\}$, the bounds in (\ref{eq_bounds_lower-upper-resolvent-diagonal}) imply that for each compact subset $A \subseteq \{\delta_1 \leq \Im(z) \leq \delta_2\}$, there thus exists a Lipschitz constant $C$ such that
\begin{align}
\left\vert \right. & \left. \log G_{ \mathbb{B} \setminus \gamma([0, k-1])   }(z; \gamma(k), \gamma(k))   - \log G_{ L+\mathfrak{e}(L); \mathbb{B} \setminus \gamma([0, k-1])   }(z; \gamma(k), \gamma(k))    \right\vert \nonumber \\
& \leq C \left\vert G_{ \mathbb{B} \setminus \gamma([0, k-1])   }(z; \gamma(k), \gamma(k))   - G_{ L+\mathfrak{e}(L); \mathbb{B} \setminus \gamma([0, k-1])   }(z; \gamma(k), \gamma(k))    \right\vert ~\mbox{,} \label{eq_proof_thm_analytic-miracles_LE_Lipschitz-bound}
\end{align}
for all $z\in A$. In particular, (\ref{eq_proof_thm_analytic-miracles_LE_Lipschitz-bound}) shows that $\{g_L ~,~ L \in \mathbb{N} \}$ forms a normal family of analytic functions on $\{ \Im z > 0\}$.

Similar to part (i) of the proof, an argument based on the RW-expansion results in the upper bound in (\ref{eq_thm_analytic-miracles_keybound}) for the absolute value on the right-hand side of (\ref{eq_proof_thm_analytic-miracles_LE_Lipschitz-bound}). Thus, we can conclude that (\ref{eq_proof_thm_analytic-miracles_LE-reduction}) holds for $z \in \mathbb{C}$ with $\Im z > K$, which, as before, extends to all of $\{\Im z > 0\}$ by Vitali's convergence theorem. 
\end{itemize}
\end{proof}

In view of our later applications to the density of states measure, part (i) of Theorem \ref{thm_analytic-miracles} has an immediate consequence formulated below in Corollary \ref{coro_DOSm-wek-conv}. Here, it is useful to observe that the left-hand side of (\ref{eq_thm_analytic-miracles_DOS}) can be rewritten in terms of traces in the form,
\begin{align}
\frac{1}{\vert \mathcal{A}_L \vert} \mathrm{tr}( P_{\mathcal{A}_L} G(z; x, x) ) & = \frac{1}{\vert \mathcal{A}_L \vert} \sum_{x \in \mathcal{A}_L} G(z; x,x) ~\mbox{,} \\
\frac{1}{\vert \mathcal{A}_L \vert} \mathrm{tr}( P_{\mathcal{A}_L} G_{L+ \mathfrak{e}(L)}(z; x, x) ) & = \frac{1}{\vert \mathcal{A}_L \vert} \sum_{x \in \mathcal{A}_L} G_{L + \mathfrak{e}(L)}(z;x,x) ~\mbox{.}
\end{align}
Following, $\mathcal{C}_0(\mathbb{R}; \mathbb{C})$ denotes the space of complex-valued continuous functions on $\mathbb{R}$ vanishing at infinity and $\mathcal{C}_c^\infty(\mathbb{R}; \mathbb{C})$ the complex-valued and compactly supported smooth functions on $\mathbb{R}$.
\begin{corollary} \label{coro_DOSm-wek-conv}
For every $f \in \mathcal{C}_0(\mathbb{R}; \mathbb{C})$, we have
\begin{equation} \label{eq_coro_DOSm-wek-conv_claim}
\lim_{L \to \infty} \frac{1}{\vert \mathcal{A}_L \vert}  \left[ \mathrm{tr}( P_{\mathcal{A}_L}  f(H) ) -  \mathrm{tr}( P_{\mathcal{A}_L}  f(H_{L+\mathfrak{e}(L)}) )  \right] = 0 ~\mbox{.}
\end{equation}
\end{corollary}
To keep the paper self-contained we include a short proof of Corollary \ref{coro_DOSm-wek-conv} below. Our proof uses the Helffer-Sj\"ostrand functional calculus which gives an explicit formula for functions of a self-adjoint operator in terms of its resolvent. To this end, given a smooth function of compact support $f \in \mathcal{C}_c^\infty(\mathbb{R}; \mathbb{C})$, we let $\widetilde{f}: \mathbb{R}^2 \to \mathbb{C}$ denote an almost-analytic extension of $f$ which is compactly supported on $\mathbb{R}^2$ and satisfies 
\begin{align} \label{eq_almostanalytic}
\widetilde{f}(x,0) &= f(x) ~\mbox{, for all } x \in \mathbb{R} ~\mbox{, } \nonumber \\
\left\vert \partial_{\overline{z}} \widetilde{f}(x,y) \right\vert & \leq C_f \vert y \vert ~\mbox{, for all } x \in \mathbb{R} ~\mbox{, } 0 < \vert y \vert \leq 1 ~\mbox{,}
\end{align}
where $C_f$ is a constant depending on the function $f$ and, as common, $\partial_{\overline{z}} = \frac{1}{2} ( \partial_x + i \partial_y )$. Given such an almost analytic extension of $f$, it is well known that for every self-adjoint operator $H$, one has the following explicit expression of $f(H)$ in terms of the resolvent:
\begin{equation} \label{eq_HelfferS}
f(H) = \dfrac{1}{\pi} \int_{-\infty}^{+\infty} \int_{-\infty}^{+\infty} \partial_{\overline{z}} \widetilde{f}(x,y)~ G(z) ~\ud x \ud y ~\mbox{.}
\end{equation}
For existence of such almost analytic extensions and the Helffer-Sj\"ostrand functional calculus, we refer the reader e.g. to \cite{davies1}; see Chapter 2 therein.
\begin{proof}
First observe that by density, it suffices to prove the claim for $f \in \mathcal{C}_c^\infty(\mathbb{R}; \mathbb{C})$ since for every bounded self-adjoint operator $A$ on $\ell^2(\mathbb{B}; \mathbb{C})$, the continuous functional calculus implies
\begin{equation}
\left\vert \frac{1}{\vert \mathcal{A}_L \vert}  \mathrm{tr}( P_{\mathcal{A}_L}  ~g(A) ) \right\vert \leq \Vert g \Vert_\infty ~\mbox{, }
\end{equation}
for all $L \in \mathbb{N}$ and all $g \in \mathcal{C}_c(\mathbb{R}; \mathbb{C})$. Given $f \in \mathcal{C}_c^\infty(\mathbb{R}; \mathbb{C})$, using (\ref{eq_HelfferS}), we can represent
\begin{equation} \label{eq_coro_DOSm-wek-conv_HS}
 \frac{1}{\vert \mathcal{A}_L \vert} \mathrm{tr}(P_{\mathcal{A}_L} f(H_{L+\mathfrak{e}(L)})) = \dfrac{1}{\pi} \int_{-\infty}^{+\infty} \int_{-\infty}^{+\infty} \partial_{\overline{z}} \widetilde{f}(x,y) ~\frac{1}{\vert \mathcal{A}_L \vert}    \mathrm{tr}(P_{\mathcal{A}_L}    G_{L+\mathfrak{e}(L)}(z) ) ~\ud x \ud y ~\mbox{,}
\end{equation}
for each $L \in \mathbb{N}$ and similarly for $H_{L+\mathfrak{e}(L)}$ replaced by $H$. 

Notice that by (\ref{eq_almostanalytic}), we have
\begin{equation} \label{eq_coro_BC}
\left\vert \partial_{\overline{z}} \widetilde{f}(x,y) \frac{1}{\vert \mathcal{A}_L \vert}    \mathrm{tr}(P_{\mathcal{A}_L}  G_{L + \mathfrak{e}(L)}(z) ) \right\vert \leq C_f ~\mbox{, for $0 < \vert y \vert \leq 1$ and all $x \in \mathbb{R}$ .}
\end{equation}
Thus, since the integrand in (\ref{eq_coro_DOSm-wek-conv_HS}) is compactly supported, the claim in (\ref{eq_coro_DOSm-wek-conv_claim}) follows from part (i) of Theorem \ref{thm_analytic-miracles} and the dominated convergence theorem. Here, we observe that (\ref{eq_thm_analytic-miracles_DOS}) and (\ref{eq_coro_BC}) only exclude the real axis which is of zero Lebesgue measure in $\mathbb{R}^2$.
\end{proof}


\section{The Lyapunov Exponent for Ergodic Schr\"odinger Operators on the Bethe Lattice} \label{sec_dos_LE1}
\setcounter{equation}{0}
We start by introducing ergodic Schr\"odinger operators on $\mathbb{B}$ in Section \ref{sec_erogic-structure}. Underlying this definition will be the explicit representation of vertices through the action of the graph isomorphisms $\tau_1$ and $\tau_2$ given in Proposition \ref{prop_representation_vertices} which, for each $x \in \mathcal{V}$, gave rise to the generalized shift $\tau_x$ defined in (\ref{eq_generalized-shift-vertex}). Section \ref{subsec:lyapunov1} will then utilize this ergodic structure to encode the information about paths and thereby examine the limiting behavior of the exponential decay of the Green function (related to the Lyapunov exponent) and of averages of spectral quantities $f(H)$, for $f \in \mathcal{C}_c(\mathbb{R}; \mathbb{C})$, over finite path segments (related to the density of states measure). Both will form the key ingredients in our proof of the modified Thouless formula which we present in Section \ref{sec:thouless1}.

\subsection{Ergodic Structure and the Density of States} \label{sec_erogic-structure}
We let $(\Omega, \mathcal{F}, \mathbb{P})$ denote a fixed probability space. Our hypotheses for the definition of ergodic Schr\"odinger operators on the Bethe lattice are as follows: 

\medskip

\begin{itemize}
\item[(H1)] \textit{We assume that there exist two measure preserving and invertible transformations  $T_1, T_2$ on $(\Omega, \mathcal{F}, \mathbb{P}),$ and that $T_1$ is also ergodic. Given an arbitrary vertex $(0) \neq x \in \mathcal{V}$, we use the generalized shift $\tau_x$ on $\mathbb{B}$, introduced in (\ref{eq_generalized-shift-vertex}), to define an associated measure preserving and invertible transformation on $(\Omega, \mathcal{F}, \mathbb{P})$ given by
\begin{equation} \label{eq_transformations-1}
T_x: \Omega \to \Omega ~\mbox{, } T_x:=  (T_2^{d_1} T_1) (T_2^{d_2} T_1) \dots (T_2^{d_{\ell}} T_1) ~\mbox{,}
\end{equation}
where the exponents $d_1, \dots, d_\ell$ are determined by (\ref{eq_prop_representation_formulas}). If $x = (0)$, we let $T_{(0)}$ denote the identity map on $(\Omega, \mathcal{F}, \mathbb{P})$.}
\end{itemize}

\medskip

We note that our assumptions on the maps $T_1$ and $T_2$ imply the family of maps $\mathcal{T}:=\{T_x ~,~ x \in \mathcal{V} \}$ forms an {\em{ergodic family}}. Here, we recall that $\mathcal{T}$ is called an ergodic family if 1.) $\mathcal{T}$ is a measure preserving family (i.e., every element of $\mathcal{T}$ is a measure preserving map), and 2.) if every measurable set invariant under {\em{all}} elements of $\mathcal{T}$ necessarily has probability zero or one. Since we assumed that $T_1$ is an ergodic map, the family $\mathcal{T}$ forms an ergodic family.

\medskip

\begin{itemize}
    \item[(H2)] \textit{We assume that the family $\mathcal{T}$ satisfies the following {\em{consistency relations}}, which relate the elements of $\mathcal{T}$ via the family of generalized shifts on $\mathbb{B}$ by:}
\begin{equation} \label{eq_hyp_erog-Schr_consistency}
T_x T_y = T_{\tau_x(y)} ~\mbox{, for all $x,y \in \mathcal{V}$ .}
\end{equation}
\end{itemize}

\medskip
The consistency relations required by (\ref{eq_hyp_erog-Schr_consistency}) mirror the compositional property for the family of generalized shifts $\{\tau_x ~,~ x \in \mathcal{V}\}$ in (\ref{eq_compositional-law}). We note that assuming (\ref{eq_hyp_erog-Schr_consistency}) automatically implies that the following relations involving inverses also hold true:
\begin{equation} \label{eq_hyp_erog-Schr_consistency-1}
 T_x^{-1} T_y = T_{\tau_x^{-1}(y)} ~\mbox{, for all $x,y \in \mathcal{V}$ .}
 \end{equation}
Indeed, (\ref{eq_hyp_erog-Schr_consistency-1}) parallels the compositional property involving inverses for the generalized shifts in (\ref{eq_compositional-law-1}); the proof of (\ref{eq_hyp_erog-Schr_consistency-1}) using (\ref{eq_hyp_erog-Schr_consistency}) follows from an analogous argument than what we used earlier to establish (\ref{eq_compositional-law-1}) from (\ref{eq_compositional-law}).

\medskip


\begin{itemize} 
\item[(H3)] \textit{We assume that $v \in L^\infty(\Omega; \mathbb{R})$ is a real-valued function which will serve as the {\em{sampling function}} which generates the potential of the ergodic Schr\"odinger operator. Given the ergodic family $\mathcal{T}$ and the function $v$, for each realization $\omega \in \Omega$, we define the potential} $V_\omega \in \ell^\infty(\mathbb{B}; \mathbb{R})$ by
\begin{equation} \label{eq_ergodic-pot}
[V_\omega](x) := v(T_x^{-1} ~\omega) ~\mbox{.}
\end{equation}
\end{itemize}
\medskip

Assuming the set-up described in hypotheses (H1)--(H3), we define the associated ergodic Schr\"odinger operator on $\mathbb{B}$ as the family of operators
\begin{equation} \label{eq_ergodic_SOP}
H_\omega = \Delta_\mathbb{B} + V_\omega ~\mbox{, }
\end{equation}
whose elements are indexed by the possible realizations $\omega \in \Omega$. Crucial for the collective spectral properties of ergodic Schr\"odinger operators are the {\em{covariance relations}} which will relate realizations of the operators $H_\omega$ via the action of the elements of the family $\mathcal{T}$. We note that the covariance relations rely on the consistency relations in (\ref{eq_hyp_erog-Schr_consistency}) of hypothesis (H2). 

To formulate the covariance relations for ergodic Schr\"odinger operators on $\mathbb{B}$, we first use the family of generalized shifts $\{\tau_x ~,~ x \in \mathcal{V}\}$ to define the associated family of unitary operators $U_x: \ell^2(\mathbb{B}; \mathbb{C}) \to \ell^2(\mathbb{B}; \mathbb{C})$, for $x \in \mathcal{V}$, by 
\begin{equation} \label{eq_unitaries}
[U_x \psi](y) = \psi(\tau_x^{-1} ~y) ~\mbox{, $y \in \mathcal{V}$ .}
\end{equation}
In particular, the definition in (\ref{eq_unitaries}) implies that for each vertex $x \in \mathcal{V}$, one has 
\begin{equation}
U_x \delta_{(0)} = \delta_x ~\mbox{.}
\end{equation}
Observe that the definition of the Laplacian on $\mathbb{B}$ in (\ref{eq_laplacian}) implies the commutator
\begin{equation} \label{eq_commutator_unitary-laplacian}
\Delta_\mathbb{B} U_x - U_x \Delta_\mathbb{B} = 0 ~\mbox{, for each $x \in \mathcal{V}$ .}
\end{equation}
 
Using the definition of the potential in (\ref{eq_ergodic-pot}) and the consistency relations in (\ref{eq_hyp_erog-Schr_consistency}), it is straight forward to check that 
\begin{equation}
U_x V_\omega U_x^{-1} = V_{T_x \omega} ~\mbox{, for each $x \in \mathcal{V}$ and $\omega \in \Omega$ ,}
\end{equation}
which, by (\ref{eq_commutator_unitary-laplacian}), results in the desired covariance relations for ergodic Schr\"odinger operators on $\mathbb{B}$, given by
\begin{equation}  \label{eq_covariance}
U_x H_\omega U_x^{-1} = H_{T_x \omega} ~\mbox{, for each $x \in \mathcal{V}$ and $\omega \in \Omega$ .}
\end{equation}
These covariance relations guarantee the existence of the deterministic spectrum $\Sigma \subset \mathbb{R}$ for the family of ergodic \Schr operators $H_\omega$, a result due to Pastur \cite{pastur1}, see also \cite[chapter 9]{CyconFroeseKirschSimon_book_1987}.

For later purposes, we observe that by the continuous functional calculus, (\ref{eq_covariance}) extends to all functions $f \in \mathcal{C}_0(\mathbb{R};\mathbb{C})$ vanishing at infinity:
\begin{equation} \label{eq_covariance-relation}
U_x f(H_\omega) U_x^{-1} = f(H_{T_x \omega}) ~\mbox{, for each $x \in \mathcal{V}$, $\omega \in \Omega$ .}
\end{equation}

{\em{Random Sch\"odinger operators with independent identically distributed (iid) potential}} sequence forms an important special case of this ergodic framework. For this special case, we let $\mu$ be a fixed Borel probability measure on $\mathbb{R}$ supported on $[-C,C]$, for some $0 < C < +\infty$. The measure $\mu$ is usually referred to as the single-site probability measure. As a probability space, we consider the sequence space 
\begin{equation} \label{eq_random_sequence-space}
\Omega = [-C,C]^\mathcal{V} ~\mbox{, for a fixed constance $C>0$ ,}
\end{equation}
equipped with the usual cylinder sigma algebra $\mathcal{F}$ and the product measure
\begin{equation}
\mathbb{P} = \bigotimes_{\mathcal{V}} \mu ~\mbox{.}
\end{equation}

We introduce the maps $T_1, T_2$ on $(\Omega, \mathcal{F}, \mathbb{P})$, required by hypothesis (H1), as follows: for a sequence $\omega=(\omega(x))_{x \in \mathcal{V}} \in \Omega$ and $1 \leq j \leq 2$, we let
\begin{equation}
[T_j \omega](x) = \omega(\tau_j^{-1} x) ~\mbox{, for each $x \in \mathcal{V}$ .}
\end{equation}
By the definition of the cylinder sigma algebra, $T_1$ and $T_2$ are both measure preserving. Moreover, since $\tau_1$ is a generalized translation between the levels of $\mathbb{B}$, $T_1$ is strongly mixing and hence ergodic. In particular, the family $\{T_x ~,~ x \in \mathcal{V} \}$ constitutes an ergodic family as required by hypothesis (H1).

For the sampling function $v$ that generates the random potential according to (\ref{eq_ergodic-pot}) of hypothesis (H3), we consider the evaluation map
\begin{equation}
v: [-C,C]^\mathcal{V} \to [-C,C]^\mathcal{V} ~\mbox{, }  v(\omega) = \omega(0) ~\mbox{,}
\end{equation}
in which case, (\ref{eq_ergodic-pot}) implies
\begin{equation}
[V_\omega](x) = \omega(x) ~\mbox{, for all } x \in \mathcal{V} ~\mbox{.}
\end{equation}

Finally, the consistency relations (\ref{eq_hyp_erog-Schr_consistency}) required by hypothesis (H2) follow directly from the compositional properties of the shift maps in (\ref{eq_compositional-law})--(\ref{eq_compositional-law-1}).  Indeed, using (\ref{eq_compositional-law}), we have for each $\omega \in \Omega$ and $x,y,z \in \mathcal{V}$ that
\begin{align}
[T_x T_y \omega](z) & = [T_y \omega](\tau_x^{-1}(z)) = \omega((\tau_y^{-1} \circ \tau_x^{-1})(z)) = \omega( (\tau_x \circ \tau_y)^{-1}(z)) \nonumber \\
   & = \omega( \tau_{\tau_x(y)}^{-1}(z) ) = [T_{\tau_x(y)} \omega](z) ~\mbox{.}
\end{align}

To conclude this section, we define the {\em{density of states measure}} associated with the ergodic operator (\ref{eq_ergodic_SOP}) and the underlying sampling function $v$ as the unique (compactly supported) Borel measure $\ud n_v$ on $\mathbb{R}$ such that, for all $f \in \mathcal{C}_c(\mathbb{R}; \mathbb{C})$, one has
\begin{equation} \label{eq_dosm_defn}
\int f(E) ~\ud n_v(E) = \mathbb{E} \langle \delta_{(0)} , f(H_\omega) \delta_{(0)} \rangle ~\mbox{.}
\end{equation}
Observe that the covariance relations for functions of the operator in (\ref{eq_covariance-relation}) imply that 
\begin{equation}
\int f(E) ~\ud n_v(E) = \mathbb{E} \langle \delta_{x} , f(H_\omega) \delta_{x} \rangle ~\mbox{, for all $x \in \mathcal{V}$ ,}
\end{equation}
so that, for a full measure set of realizations $\omega \in \Omega$, one has
\begin{equation}\label{det_spectrum1}
\Sigma:= \mathrm{supp} (\ud n_v) = \sigma(H_\omega) ~\mbox{,}
\end{equation}
so that $\Sigma$ is the deterministic spectrum of the ergodic Schr\"odinger operator $H_\omega$.



\subsection{Exponential Decay of the Green Function and the Lyapunov Exponent} \label{subsec:lyapunov1}

In analogy to the theory of discrete Schr\"odinger operators on $\mathbb{Z}$, given a deterministic Schr\"odinger operator $H$ on $\mathbb{B}$ with connectivity $\kappa$, we define the collection of $\kappa + 1$ {\em{Weyl $m$-functions}} as follows: for $0 \leq j \leq \kappa $ and $z \in \mathbb{C}\setminus \mathbb{R}$, we let
\begin{equation}\label{eq:m_fnc1}
{M}_{j}(z) := \langle \delta_{y}, (H_y^{+; (0)} - z)^{-1} \delta_y \rangle ~\mbox{, for $y = (0,j)$ ,}
\end{equation}
where $H_y^{+; (0)}$ is associated with the rooted tree with root at $y$, obtained by decoupling $H$ at $(0)$; see (\ref{eq_rooted_hamiltonian_zero}). 

For an ergodic Schr\"odinger operator $H_\omega$ on $\mathbb{B}$, we use the notation ${M}_{j; \omega}(z)$ for the $\kappa + 1$ Weyl $m$-functions, and define the Lyapunov exponent as follows:
\begin{definition}\label{defn:lyapunov_defn1}
Given an ergodic Schr\"odinger operator on $\mathbb{B}$ satisfying (H1)--(H3), we define its {\em{Lyapunov exponent}} by 
\begin{equation} \label{eq_defn_LyapExp}
 \mathcal{L}(z): = - \mathbb{E} \{ \log \vert {M}_{0; \omega} (z) \vert \} = - \mathbb{E} \{ \log \vert G_{0; \omega} (z; y,y) \vert\} ~ ~\mbox{, for $y = (0,0)$ .}
\end{equation}
\end{definition}
Observe that for $0 \leq j \leq \kappa$, one has 
\begin{equation}
\delta_{(0,j)} = \delta_{\tau_2^{j}(0,0)} ~\mbox{, }
\end{equation}
so that from the covariance relation in (\ref{eq_covariance-relation}) and since $T_2$ is measure preserving, we can conclude that
\bea\label{eq_defn_LyapExp-rot}
\mathbb{E} \{  \log \vert {M}_{j; \omega} (z) \vert \} &  = & \mathbb{E} \{ \log \vert \langle \delta_{(0,0)}, (H_{(0,0);T_2^{-j} \omega} ^{+; (0)} - z)^{-1} \delta_{(0,0)} \rangle \vert \}
 \nonumber \\
  & = & \mathbb{E} \{ \log \vert {M}_{0; \omega} (z) \vert \}  \nonumber \\
   & =  & - \mathcal{L}(z) ~\mbox{.}
\eea
This shows that the Lyapunov exponent does not depend on which of the $(\kappa + 1)$ Weyl $m$-function is used in Definition \ref{defn:lyapunov_defn1}.

For discrete Sch\"rodinger operators on $\mathbb{Z}$ (or equivalently, Schr\"odinger operators on $\mathbb{B}$ with $\kappa = 1$) and $z \in \mathbb{C}\setminus \mathbb{R}$, (\ref{eq_defn_LyapExp}) is indeed equivalent to the more common definition of the Lyapunov exponent using transfer matrices; see e.g. \cite{CyconFroeseKirschSimon_book_1987} or \cite{damanik1}. For ergodic Schr\"odinger operators on $\mathbb{Z}$, the Lyapunov exponent plays a central role in proofs of localization, in which case it is associated with the exponential decay-rate of the Green function. Generalizing to the Bethe lattice ($\kappa \geq 2$), given an infinite path $\gamma: \mathbb{N}_0 \to \mathbb{B}$ with $\gamma(0) = (0)$, one would therefore expect that for a typical realization $\omega$, 
\begin{equation} \label{eq_path-limit}
\dfrac{1}{L} \log \left\vert G_\omega(z; \gamma(0), \gamma(L-1) ) \right\vert \to - \mathcal{L}(z) ~\mbox{, as $L \to \infty$ .}
\end{equation}

Indeed, for the case of a random Schr\"odinger operator with iid potentials as considered in Section \ref{sec_erogic-structure}, Aizenman and Warzel prove that for each fixed path there exists a full measure set of realizations $\omega \in \Omega$, so that (\ref{eq_path-limit}) holds \cite[Theorem 5.2]{aw_JEMS_2013}. We note that in contrast to ergodic Schr\"odinger operators on $\mathbb{Z}$, the claimed relation in (\ref{eq_path-limit}) for ergodic Schr\"odinger operators on $\mathbb{B}$ does not merely follow from a suitable version of the ergodic theorem. In \cite{aw_JEMS_2013}, the authors use large deviation estimates \cite[Theorem B.1]{aw_JEMS_2013} to prove the almost sure convergence in (\ref{eq_path-limit}).

To appreciate the complication when analyzing limits along paths of the type (\ref{eq_path-limit}), observe that (\ref{eq_infinite-path-labeling}) and the explicit representation of vertices in Proposition \ref{prop_representation_vertices} sets up a one-to-one correspondence between infinite paths $\gamma: \mathbb{N}_0 \to \mathbb{B}$ with $\gamma(0) = (0)$ and infinite sequences with elements chosen from the $(\kappa+1)$ automorphisms on $\mathbb{B}$,
\begin{equation} \label{eq_path-autom}
\sigma_j: = \tau_2^j \tau_1 ~\mbox{, } 0 \leq j \leq \kappa ~\mbox{,}
\end{equation}
so that for each $\ell \in \mathbb{N}$, one has
\begin{equation} \label{eq_infinite-path-labeling-1}
\gamma(\ell) = \sigma_{d_1} \dots  \sigma_{d_\ell} (0) ~\mbox{,}
\end{equation}
where the exponents
\begin{equation} \label{eq_infinite-path-labeling-2}
d_1 \in \{0, \dots, \kappa\} ~\mbox{, } d_j \in \{0, \dots, \kappa - 1\} ~\mbox{, for all $2 \leq j \leq \ell$ ,}
\end{equation}
stabilize from left to right and are explicitly determined by (\ref{eq_prop_representation_formulas}). Given the definition of the ergodic family $\{T_x ~,~ x \in \mathcal{V} \}$ in (\ref{eq_transformations-1}), we can associate with each of the $(\kappa + 1)$ automorphisms in (\ref{eq_path-autom}) a measure preserving map 
\begin{equation}
S_j : = T_2^j T_1 ~\mbox{, } 0 \leq j \leq \kappa ~\mbox{.}
\end{equation}
Thus, letting
\begin{equation}
\mathcal{C}_n:=\{ S_j ~:~ 0 \leq j \leq n-1\} ~\mbox{, } n \in \mathbb{N} ~\mbox{,}
\end{equation}
we obtain a one-to-one correspondence which associates with each infinite path $\gamma: \mathbb{N}_0 \to \mathbb{B}$ satisfying $\gamma(0) = (0)$ an element of the sequence space
\begin{equation}
\mathfrak{S}_\kappa:= \mathcal{C}_{\kappa+1} \times \left\{ \prod_{2 \leq j} \mathcal{C}_{\kappa} \right\}  ~\mbox{.}
\end{equation}

This point of view of infinite paths allows us the express the right hand side of (\ref{eq_path-limit}) as follows: If $\gamma: \mathbb{N}_0 \to \mathcal{V}$, $\gamma(0) = (0)$, is an infinite path associated with the sequence
$(S_{d_j})_{j \in \mathbb{N}} \in \mathfrak{S}_\kappa$ so that (\ref{eq_infinite-path-labeling-1})--(\ref{eq_infinite-path-labeling-2}) hold, then the SAW-expansion for the off-diagonal elements of the Green function yields
\begin{align}
\dfrac{1}{L} \log \left\vert G_\omega(z; \gamma(0), \gamma(L-1) ) \right\vert & = \dfrac{1}{L} \log \left\vert  G_\omega(z; (0), (0) ) \right\vert \nonumber \\
 & + \dfrac{1}{L} \sum_{j=1}^{L-1} \log \left\vert \left\langle \delta_{\gamma(j)}, G_{\mathbb{B} \setminus \gamma([0,j-1]); \omega}(z)  \delta_{\gamma(j)} \right\rangle \right\vert ~\mbox{,} 
\label{eq_prop_lyap_through_green_saw}
\end{align}
for each fixed $L \in \mathbb{N}$ and $z \in \mathbb{C} \setminus \mathbb{R}$. Since Schr\"odinger operators (\ref{eq_ergodic_SOP}) only couple {\em{adjacent}} vertices, the tree structure of $\mathbb{B}$, the covariance relations for ergodic Schr\"odinger operators in (\ref{eq_covariance-relation}), and (\ref{eq_rooted_hamiltonian_zero}) imply
\begin{align}
\left\langle \delta_{\gamma(j)}, \right. & \left. G_{\mathbb{B} \setminus \gamma([0,j-1]); \omega}(z)  \delta_{\gamma(j)} \right\rangle = \left\langle \delta_{\gamma(j)}, G_{\mathbb{B} \setminus \{\gamma(j-1)\}); \omega}(z)  \delta_{\gamma(j)} \right\rangle \nonumber \\
& = \left\langle \delta_{\gamma(1)}, G_{\mathbb{B} \setminus \{0\}; T_{\gamma(j)}^{-1} \omega}(z)  \delta_{\gamma(1)} \right\rangle = M_{a_1; S_{d_j}^{-1} \dots S_{d_1}^{-1} \omega}(z) ~\mbox{,} \label{eq_prop_lyap_through_green_covariance}
\end{align}
for all $1 \leq j \leq L-1$. The function $M_{j;\omega}(z)$ is defined 
in \eqref{eq_defn_LyapExp-rot}, and $a_1 \in \{0, \dots, \kappa\}$ in (\ref{eq_prop_lyap_through_green_covariance}) is determined by the fixed path $\gamma$ under consideration such that
\begin{equation} \label{eq_prop_lyap_through_green_covariance-1}
\gamma(1) = (0,a_1) ~\mbox{.}
\end{equation}

In particular, for each fixed $z \in \mathbb{C} \setminus \mathbb{R}$, using \eqref{eq_prop_lyap_through_green_saw} and \eqref{eq_prop_lyap_through_green_covariance},
 we can represent the exponential decay rate of the Green function along the given path $\gamma$---expressed by the left-hand side of (\ref{eq_path-limit})---in the form:
\begin{align} \label{eq_MRP-motivation-green-fc-decay}
\dfrac{1}{L} \log \left\vert G_\omega(z; \gamma(0), \gamma(L-1) ) \right\vert & = o(1) +  \dfrac{1}{L} \sum_{j=1}^{L-1} \log \left\vert M_{a_1; S_{d_j}^{-1} \dots S_{d_1}^{-1} \omega}(z) \right\vert ~\mbox{, }
\end{align}
as $L \to \infty$ (for each $\omega$ in the sense of a limit superior or inferior), so that (\ref{eq_path-limit}) finally becomes
\begin{equation} \label{eq_MRP-motivation-green-fc-decay}
\dfrac{1}{L} \sum_{j=1}^{L} \log \left\vert M_{a_1; S_{d_j}^{-1} \dots S_{d_1}^{-1} \omega}(z) \right\vert \to - \mathcal{L}(z) ~\mbox{, as $L \to \infty$ .}
\end{equation}

As mentioned earlier, for the case of random Schr\"odinger operators with iid potentials, it was proven in \cite{aw_JEMS_2013} that the limit in (\ref{eq_path-limit}) exists for a full measure set of realizations $\omega \in \Omega$; in particular, the same thus holds true in this case for the limit in (\ref{eq_MRP-motivation-green-fc-decay}). The expression of the Lyapunov exponent in (\ref{eq_MRP-motivation-green-fc-decay}) motivates the following terminology for general ergodic Schr\"odinger operators satisfying the set-up described in the hypotheses (H1)--(H3) of Section \ref{sec_erogic-structure}:
\begin{definition} \label{def_MRP}
Let $\gamma: \mathbb{N}_0 \to \mathcal{V}$, $\gamma(0) = (0)$, be an infinite path associated with the sequence
\begin{equation}
(S_{d_j})_{j \in \mathbb{N}} \in \mathfrak{S}_\kappa ~\mbox{}
\end{equation}
so that (\ref{eq_infinite-path-labeling-1})--(\ref{eq_infinite-path-labeling-2}) hold. We say that $\gamma$ is {\em{macroscopically representative}} (MRP) if for every $f \in L^1 \cap L^\infty(\Omega, \mathcal{F}, \mathbb{P})$, there exists a full measure set $\Omega_0(f) \subseteq \Omega$ such that for all $\omega \in \Omega_0(f)$, one has
\begin{equation} \label{eq_def_MRP-limit}
\dfrac{1}{L} \sum_{j=1}^{L} f( (S_{d_1} \dots S_{d_j})^{-1} \omega) = \dfrac{1}{L} \sum_{j=1}^{L} f( S_{d_j}^{-1} \dots S_{d_1}^{-1} \omega)  \to \mathbb{E}(f) ~\mbox{, as $L \to \infty$ .}
\end{equation}
\end{definition}
To give examples which illustrate the existence of MRP paths, we note that every path $\gamma$ associated with a sequence $(S_{d_j})_{j \in \mathbb{N}} \in \mathfrak{S}_\kappa$ for which
\begin{equation} \label{eq_MRP-existence}
d_j = 0 ~\mbox{, eventually in $j \in \mathbb{N}$ ,}
\end{equation}
is necessarily MRP. This follows since (\ref{eq_MRP-existence}) implies that eventually in $j \in \mathbb{N}$, $S_j = T_1$, so that the ergodicity of $T_1$ implies that (\ref{eq_def_MRP-limit}) holds by Birkhoff's ergodic theorem.

We end this section with the following remark, which puts Definition \ref{def_MRP} in the context of our discussion of the Thouless formula in Section \ref{sec:thouless1}:
\begin{remark} \label{remark_MRP-defn-applic}
Suppose that $\gamma$ is a fixed MRP path. In our proof of the Thouless formula in Section \ref{sec:thouless1}, we will apply Definition \ref{def_MRP} for the following two functions: 
\begin{enumerate}
\item {\bf{Relation to the Lyapunov exponent:}} For each $a_1 \in \{0, \dots, \kappa \}$ (uniquely determined by the path at hand through (\ref{eq_prop_lyap_through_green_covariance-1})) and $z \in \mathbb{C} \setminus \mathbb{R}$, letting
\begin{equation}
f_1(\omega) = \log \left\vert M_{a_1; \omega}(z) \right\vert ~\mbox{,}
\end{equation}
the limit in (\ref{eq_def_MRP-limit}) with $f=f_1$ recovers our expression of the Lyapunov exponent in (\ref{eq_defn_LyapExp-rot}). 

\item {\bf{Relation to the density of states measure:}} Fixing an arbitrary $z \in \mathbb{C} \setminus \mathbb{R}$, denote
\begin{equation*}
g(E) = \log \vert E- z \vert ~\mbox{, for } E \in \mathbb{R} ~\mbox{,}
\end{equation*}
and consider the function
\begin{equation}
f_2(\omega) = \left\langle \delta_{(0)}, g(H_\omega) \delta_{(0)} \right\rangle ~\mbox{.}
\end{equation}
Since $g \in \mathcal{C}_0(\mathbb{R}; \mathbb{C})$, we know that $f \in L^1 \cap L^\infty(\Omega, \mathcal{F}, \mathbb{P})$. Thus, if $\gamma$ is a fixed MRP path, combining (\ref{eq_def_MRP-limit}) with the definition of the density of states measure in (\ref{eq_dosm_defn}), we obtain
\begin{equation}
\dfrac{1}{L} \sum_{j=1}^{L} f_2( S_{d_1}^{-1} \dots S_{d_j}^{-1} \omega)  \to \int \log \vert E- z \vert ~\ud n_v(E) ~\mbox{, as $L \to \infty$ ,}
\end{equation}
for all realizations $\omega$ in the full measure set $\Omega(f_2)$.
\end{enumerate}
\end{remark}


\section{A Modified Thouless Formula for the Bethe Lattice}\label{sec:thouless1}
\setcounter{equation}{0}


This section contains our main result, the modified Thouless formula for the Bethe lattice that consists of a Thouless-like term and a nontrivial remainder term. 

\medskip

\begin{theorem}\label{thm:thouless}(Modified Thouless formula) 
Let $H_\omega$ be an ergodic Schr\"odinger operator on the Bethe lattice $\mathbb{B}$ with coordination number $\kappa \geq 2$. We assume hypotheses (H1)-(H3) of section \ref{sec_dos_LE1}. Let $d n_v$ be the associated density of states measure defined in \eqref{eq_dosm_defn} with support given by the deterministic spectrum $\Sigma$. Then, the Lyapunov exponent $\mathcal{L}(z)$, defined in \eqref{eq_defn_LyapExp}, is given by 
\beq\label{eq:lyap_thouless1}
\mathcal{L}(z) = \int_\Sigma \log | z - E| ~dn_v(E) + R(z) ~\mbox{,}
\eeq
where $R(z)$ is the non-zero remainder term, depending on $\kappa$, defined in \eqref{eq_remainder-term}. Moreover, $R(z)$ satisfies the upper bound
\begin{equation} \label{eq:lyap_thouless2}
\vert R(z) \vert \leq C(z) (\kappa - 1) ~\mbox{,}
\end{equation}
for a constant $C(z)$ described in (\ref{eq_thouless-remainder-bound}).
\end{theorem}
Observe that the upper bound for the remainder term in (\ref{eq:lyap_thouless2}) shows that for ergodic Schr\"odinger operators on $\mathbb{Z}$ ($\kappa = 1$), one has $R(z) = 0$, which recovers the ``usual'' Thouless formula.

\medskip
The key ingredient in the proof is Theorem \ref{thm_analytic-miracles} which allows us to replace quantities related to the infinite-volume resolvent operator with finite-volume restrictions. When doing so, the ``price to pay'' is to enlarge the finite volume by a growing amount, expressed by a function $\mathfrak{e}: \mathbb{N} \to \mathbb{N}$ so that $\mathfrak{e}(L) \to +\infty$, as $L \to \infty$. In the following, we will take $\mathfrak{e}(L) = L$ so that the radius of the enlarged volume amounts to ball of radius $L + \mathfrak{e}(L) = 2L$. As we will show (see \eqref{eq_remainder-term}), the particular choice of the enlargement function $\mathfrak{e}$ will not affect our final conclusions. 

The proof of Theorem \ref{thm:thouless} is presented in Section \ref{sec:thouless1-proof}. In Section \ref{sec:remainder1} prove that the remainder term is nontrivial and therefore essential by showing that, for the free Laplacian on the Bethe lattice, the remainder term does not vanish if $\kappa \geq 2$ and is a function of $z \in [- 2 \sqrt{\kappa}, 2 \sqrt{\kappa}]$. 


\subsection{Proof of the Modified Thouless Formula} \label{sec:thouless1-proof}

We fix an infinite path $\gamma$, $\gamma(0) = (0)$, which is MRP. Let $z \in \mathbb{C}$ with $\Im z > 0$. We start with the expression of the Lyapunov exponent as an exponential decay rate of the off-diagonal elements of the Green function.  Combining (\ref{eq_MRP-motivation-green-fc-decay})--(\ref{eq_MRP-motivation-green-fc-decay}) with Remark \ref{remark_MRP-defn-applic} (1), for $\mathbb{P}$-almost sure (a.s.) $\omega$, we have
\begin{equation}\label{eq:lyap1}
\mathcal{L}(z) = -\lim_{L \to \infty} \dfrac{1}{L} \log \left\vert G_\omega(z; \gamma(0), \gamma(L-1) ) \right\vert . 
\end{equation}
Using part (ii) of Theorem \ref{thm_analytic-miracles}, the right side of \eqref{eq:lyap1} can be computed using finite-volume approximations so that
\begin{equation} \label{eq_thouless-Lyap-start}
\mathcal{L}(z) = - \lim_{L \to \infty} \dfrac{1}{L} \log \left\vert G_{2L; \omega}(z; \gamma(0), \gamma(L-1) ) \right\vert ~\mbox{.}
\end{equation}

Expanding the right-hand side of (\ref{eq_thouless-Lyap-start}) with aid of the SAW-expansion, we get
\begin{equation} \label{eq_thouless_finite-vol}
\left\vert G_{2L; \omega}(z; \gamma(0), \gamma(L-1) ) \right\vert = \left\vert  \prod_{j=0}^{L-1}  \left\langle \delta_{\gamma(j)}, G_{\Lambda_{2L} \setminus \gamma([0,j-1]); \omega}(z)  \delta_{\gamma(j)} \right\rangle      \right\vert ~\mbox{,}
\end{equation}
where we recall the notation for $j = 0$ that $G_{\Lambda_{2L} \setminus \gamma([0,j-1]); \omega}(z) = G_{\Lambda_{2L}, \omega}(z) $.
We denote by $C^{(i,j)}$ the $(i,j)^{\rm th}$ cofactor of a matrix $C$. Observe that for $0 \leq j \leq L$, the tree structure of the Bethe lattice and the nearest neighbor coupling of the Laplacian imply that
\begin{equation}
\left\vert \left[ H_{\Lambda_{2L} \setminus \gamma([0,j-1]); \omega     }   - z \right]^{(\gamma(j),\gamma(j))} \right\vert = \left\vert \mathrm{det}\left[ H_{\Lambda_{2L} \setminus \gamma([0,j]);\omega}- z \right] \right\vert ~\mbox{.}
\end{equation}
Thus, using cofactors to compute the matrix inverses on the right-hand side of (\ref{eq_thouless_finite-vol}), results in
\begin{align}
- \dfrac{1}{L} \log \left\vert \right. & \left. G_{2L; \omega}(z; \gamma(0), \gamma(L-1) ) \right\vert =  \dfrac{1}{L} \log \left\vert \dfrac{\mathrm{det} \left[H_{\Lambda_{2L}; \omega} - z \right]}{  \mathrm{det} \left[ H_{\Lambda_{2L} \setminus \gamma([0,L-1]);\omega} - z \right] }     \right\vert  \nonumber \\
&
= \dfrac{1}{L} \left\{ \mathrm{tr} \log \left\vert H_{\Lambda_{2L}; \omega} - z \right\vert - \mathrm{tr} \log \left\vert H_{\Lambda_{2L} \setminus \gamma([0,L-1]);\omega} - z \right\vert \right\}  \nonumber \\
& = \dfrac{1}{L} \mathrm{tr}( P_{\gamma([0, L-1])} \log \left\vert H_{\Lambda_{2L}; \omega} - z \right\vert) + R_L(\omega;z) ~\mbox{,} \label{eq_eq_thouless_trace-log}
\end{align}
where we let
\begin{equation}
R_L(\omega;z):= \dfrac{1}{L} \left[ \mathrm{tr} \left( P_{\Lambda_{2L} \setminus \gamma([0, L-1])} \log \left\vert H_{\Lambda_{2L}; \omega} - z \right\vert \right) - \mathrm{tr} \log \left\vert H_{\Lambda_{2L} \setminus \gamma([0,L-1]);\omega} - z \right\vert \right] ~\mbox{.} \label{eq_eq_thouless_trace-log-remainder}
\end{equation}
We note that, since $\Im z > 0$, the map
\begin{equation}
\mathbb{R} \ni E \mapsto \log \vert E - z \vert ~\mbox{,}
\end{equation}
is harmonic in a neighborhood of $\mathbb{R}$, whence (\ref{eq_eq_thouless_trace-log}) -- (\ref{eq_eq_thouless_trace-log-remainder}) are well-defined by functional calculus. 

Finally, we analyze the limit as $L \to \infty$ of the first term on the right in \eqref{eq_eq_thouless_trace-log-remainder}.
Combining Theorem \ref{thm_analytic-miracles} with the comments in Remark \ref{remark_MRP-defn-applic} shows that the limit $L \to \infty$ of (\ref{eq_eq_thouless_trace-log}) exists for $\mathbb{P}$-a.e. $\omega$ with
\begin{equation} \label{eq_remainder-term}
R(z): = \lim_{L \to \infty} R_L(\omega;z) = \mathcal{L}(z) - \int_{\Sigma} \log \vert E - z \vert ~\ud n_v(E) . 
\end{equation}
In particular, this limit is $\mathbb{P}$-a.s.\ independent of $\omega$. Additionally, the right-hand side of (\ref{eq_remainder-term}) shows that that the remainder term $R(z)$ is independent of our choice for the function $\mathfrak{e}(L)=2L$ in Theorem \ref{thm_analytic-miracles}. In summary, we thus obtain the following {\em{modified version of the Thouless formula}}:
\begin{equation} \label{eq_gen_Thouless}
\mathcal{L}(z) = \int_{\Sigma} \log \vert E - z \vert ~\ud n_v(E) + R(z) ~\mbox{,}
\end{equation}
where the {\em{remainder term}} $R(z)$ is defined as the a.s. limit of (\ref{eq_eq_thouless_trace-log-remainder}) as $L \to +\infty$. In section \ref{sec:remainder1}, we will prove that the remainder $R(z)$ is, in general, a non-vanishing function of $z$ for all $\kappa \geq 2$. We refer to the first term on the right in \eqref{eq_gen_Thouless} as the \textbf{Thouless-like term} since this term represents the Lyapunov exponent for ergodic Schr\"odinger operators on $\mathbb{Z}$.  

Intuitively, from (\ref{eq_eq_thouless_trace-log-remainder}), $R(z)$ expresses the averaged effects of the coupling of the path $\gamma$ to its surrounding. To make this interpretation more explicit, we conclude this section by analyzing $R(z)$ using the geometric resolvent estimate. In particular, for the case that $\kappa = 1$ (i.e., for an ergodic Schr\"odinger operator on $\mathbb{Z}$), this analysis will show that the remainder term is zero, thereby resulting in the usual version of the Thouless formula. This is consistent with our heuristic of $R(z)$ as expressing the averaged effects of the coupling of the path $\gamma$ to its surrounding, since for $\kappa = 1$, the path is coupled to its surroundings through a {\em{constant}} number of vertices (the two vertices, neighboring the endpoints of the path). In contrast, for $\kappa \geq 2$, the number of neighboring vertices to a path will grow with the length of the path. 

To analyze the remainder term, we fix $L \in \mathbb{N}$ and use Helffer-Sj\"ostrand functional calculus to represent $R_L(\omega;z)$ through the resolvent: to this end, for fixed $z \in \mathbb{C} \setminus \mathbb{R}$, we let
\begin{equation}
f_z(E) = \phi(E) \log \vert E - z \vert ~\mbox{, } x \in \mathbb{R} ~\mbox{, }
\end{equation}
where $\phi$ is a $\mathcal{C}^\infty$-bump function which satisfies
\begin{equation}
\phi(E) \equiv 1 ~\mbox{, for $E \in [- (\kappa +1) - \Vert v \Vert_\infty, (\kappa + 1) + \Vert v \Vert_\infty]$ .}
\end{equation}
Here, we note that the set $[- (\kappa +1) - \Vert v \Vert_\infty, (\kappa + 1) + \Vert v \Vert_\infty]$ is chosen so that it contains all of the spectrum of $H_{\Lambda_{2L}; \omega}$ and $H_{\Lambda_{2L} \setminus \gamma([0,L-1]);\omega}$ for all $L \in \mathbb{N}$ and $\omega \in \Omega$. Now let $\widetilde{f_z}$ be an almost-analytic extension of $f$ so that $\widetilde{f_z}: \mathbb{R}^2 \to \mathbb{C}$ is compactly supported on $\mathbb{R}^2$ and satisfies
\begin{align} \label{eq_almostanalytic-degree2}
\widetilde{f_z}(x,0) &= f_z(x) ~\mbox{, for all } x \in \mathbb{R} ~\mbox{, } \nonumber \\
\left\vert \partial_{\overline{z}} \widetilde{f_z}(x,y) \right\vert & \leq C_{f_z} \vert y \vert^{2} ~\mbox{, for all } x \in \mathbb{R} ~\mbox{, } 0 < \vert y \vert \leq 1 ~\mbox{.}
\end{align}
The Helffer-Sj\"ostrand functional calculus then allows us to represent $R_L(\omega;z)$ in the form,
\begin{align} \label{eq_claim_remainder_thouless_helffer}
R_L(\omega & ;z) = \dfrac{1}{\pi}\int_{-\infty}^{+\infty} \int_{-\infty}^{+\infty} \ud x \ud y ~\partial_{\overline{u}} \widetilde{f_z}(x,y)~\times \nonumber \\
& \dfrac{1}{L} \mathrm{tr} \left\{ P_{\Lambda_{2L} \setminus \gamma([0, L-1])}  \left[  G_{2L; \omega}(u) - \left( G_{\Lambda_{2L} \setminus \gamma([0, L-1]); \omega}(u) + G_{\gamma([0, L-1]); \omega}(u) \right)          \right] \right\}   ~\mbox{;} 
\end{align}
We note that we chose an almost analytic extension with degree 2 in (\ref{eq_almostanalytic-degree2}) to guarantee the convergence of the integrand in (\ref{eq_claim_remainder_thouless_helffer}), which is based on the bound (\ref{eq_remainder-apriori-bound}) given below. To avoid confusion, we reiterate that $z \in \mathbb{C} \setminus \mathbb{R}$ in (\ref{eq_claim_remainder_thouless_helffer}) is fixed and that we use $u:= x + i y$ as the integration variable in the double integral of (\ref{eq_claim_remainder_thouless_helffer}).

For each $u \in \mathbb{C} \setminus \mathbb{R}$, we can use the geometric resolvent identity to analyze the integrand in (\ref{eq_claim_remainder_thouless_helffer}),
\begin{equation} \label{eq_claim_remainder_thouless_helffer_integrand}
I_L(\omega; u): = \dfrac{1}{L} \mathrm{tr} \left\{ P_{\Lambda_{2L} \setminus \gamma([0, L-1])}  \left[  G_{2L; \omega}(u) - \left( G_{\Lambda_{2L} \setminus \gamma([0, L-1]); \omega}(u) + G_{\gamma([0, L-1]); \omega}(u) \right)          \right] \right\} ~\mbox{.}
\end{equation} 
For ease of notation, we write
\begin{equation}
G_\omega^{(\gamma; L)}(u):= G_{\Lambda_{2L} \setminus \gamma([0, L-1]); \omega}(u) + G_{\gamma([0, L-1]); \omega}(u) ~\mbox{.}
\end{equation}
Moreover, given a set of vertices $\mathcal{A} \subsetneq \Gamma_{L}$, we define its boundary relative to $\Gamma_{L}$ by
\begin{equation}
\partial_{L} \mathcal{A} := \{(m,l) \in \mathcal{A} \times \Gamma_L \setminus \mathcal{A} ~:~ m \sim l \} ~\mbox{,}
\end{equation}
where, as earlier, $k \sim l$ denote nearest neighboring vertices. 

Using the geometric resolvent identity we obtain the following expression for (\ref{eq_claim_remainder_thouless_helffer_integrand})
\begin{align}
I_L(\omega; u)  = \dfrac{1}{L} \sum_{j \in \Lambda_{2L} \setminus \gamma([0, L-1])} \sum_{(m,l) \in \partial_{2L} \gamma([0,L-1])} G_{2L; \omega}(u; j ,m) G_\omega^{(\gamma; L)}(u; l,j) ~\mbox{,}
\end{align}
which, by Cauchy-Schwarz, can be estimated by
\begin{align}
\vert & I_L (\omega ; u) \vert \leq \dfrac{1}{L} \sum_{(m,l) \in \partial_{2L} \gamma([0,L-1])} \Vert G_{2L; \omega}(u; . ,m) \Vert_{\ell^2} \Vert G_\omega^{(\gamma; L)}(u; l, .)  \Vert_{\ell^2(\Lambda_{2L} \setminus \gamma([0, L-1]))} \\
 & \leq \dfrac{\kappa - 1}{L \vert \Im(u) \vert} \sum_{j=1}^{L-2} \Vert G_{2L; \omega}(u; . ,\gamma(j)) \Vert_{\ell^2(\Lambda_{2L} \setminus \gamma([0, L-1]))} \nonumber \\ 
 & + \dfrac{\kappa}{L \vert \Im(u) \vert} \left\{ \Vert G_{2L; \omega}(u; . ,\gamma(0)) \Vert_{\ell^2(\Lambda_{2L} \setminus \gamma([0, L-1]))} + \Vert G_{2L; \omega}(u; . ,\gamma(L-1)) \Vert_{\ell^2(\Lambda_{2L} \setminus \gamma([0, L-1]))} \right\} \label{eq_bound_IL-1} \\
 & \leq \dfrac{\kappa + 1}{L \vert \Im(u) \vert} \sum_{j=0}^{L-1} \Vert G_{2L; \omega}(u; . ,\gamma(j)) \Vert_{\ell^2(\Lambda_{2L} \setminus \gamma([0, L-1]))} ~\mbox{,} \label{eq_bound_IL}
\end{align}

First, note that the upper bound in (\ref{eq_bound_IL}) shows that for each $\omega \in \Omega$ and all $u \in \mathbb{C} \setminus \mathbb{R}$
\begin{equation} \label{eq_remainder-apriori-bound}
I_L(\omega; u) \leq \dfrac{\kappa + 1}{\vert \Im(u) \vert^2} ~\mbox{, for all $L\in \mathbb{N}$ ,}
\end{equation}
which shows that $u \mapsto I_L(\omega; u)$ is a normal family in $L \in \mathbb{N}$ for $u \in \mathbb{C} \setminus \mathbb{R}$ so that $\lim_{L \to \infty} I_L(\omega, u)$ is necessarily finite whenever it exists. 

Second, and most noteworthy, the bound in (\ref{eq_bound_IL-1}) explicitly reveals the effects of the surroundings of the path and separates them into two contributions: the first summand contains the contributions from vertices in the interior of the path; these contributions grow with the length $L$ of the path. The second summand in (\ref{eq_bound_IL-1}) quantifies the effects of the coupling of the path to its surroundings through its endpoints and thus only contains two terms, so that
\begin{align}
\dfrac{\kappa}{L \vert \Im(u) \vert} & \left\{ \Vert G_{2L; \omega}(u; . ,\gamma(0)) \Vert_{\ell^2(\Lambda_{2L} \setminus \gamma([0, L-1]))} + \Vert G_{2L; \omega}(u; . ,\gamma(L-1)) \Vert_{\ell^2(\Lambda_{2L} \setminus \gamma([0, L-1]))} \right\} \nonumber \\
& \leq \dfrac{\kappa}{L (\Im(u))^2} \to 0 ~\mbox{, as $L \to \infty$ .} 
\end{align}
In particular, we can conclude that, whenever the limit exists, one has 
\begin{equation} \label{eq_bound-remainder-kappa-dependence}
\lim_{L \to \infty} \vert I_L (\omega ; u) \vert \leq \lim_{L \to \infty} \dfrac{\kappa - 1}{L \vert \Im(u) \vert} \sum_{j=1}^{L-2} \Vert G_{2L; \omega}(u; . ,\gamma(j)) \Vert_{\ell^2(\Lambda_{2L} \setminus \gamma([0, L-1]))} ~\mbox{,}
\end{equation}
which, reveals that all contributions to the remainder term $R(z)$ stem from the coupling of the path to its surroundings through the interior of the path. Using the trivial norm bound in (\ref{eq_bound-remainder-kappa-dependence}) 
\begin{equation}
\Vert G_{2L; \omega}(u; . ,\gamma(j)) \Vert_{\ell^2(\Lambda_{2L} \setminus \gamma([0, L-1]))} \leq \dfrac{1}{\vert \Im(u) \vert} ~\mbox{,}
\end{equation}
the representation of $R_L(z)$ in (\ref{eq_claim_remainder_thouless_helffer}) together with our choice of using an almost analytic extension of degree 2, finally reveals the claimed upper bound in (\ref{eq:lyap_thouless2}) of Theorem \ref{thm:thouless},
\begin{equation} \label{eq_thouless-remainder-bound}
\vert R(z) \vert \leq \left\{\dfrac{1}{\pi}\int_{-\infty}^{+\infty} \int_{-\infty}^{+\infty} \partial_{\overline{u}} \widetilde{f_z}(x,y) ~\ud x \ud y \right\} \cdot (\kappa - 1) ~\mbox{,}
\end{equation}
where the constant $C(z)$ is explicitly given by the first factor on the right hand side of (\ref{eq_thouless-remainder-bound}). This therefore completes the proof of heorem \ref{thm:thouless}.


\subsection{Nontrivial Remainder Term for the Free Laplacian}\label{sec:remainder1}
\setcounter{equation}{0}
In this section, we will show that the remainder term $R(z)$ in our modified Thouless formula (\ref{eq_gen_Thouless}) does, in general, not vanish when $\kappa \geq 2$. We will do so for the free Laplacian ($\Delta_\mathbb{B}$) 
for which we will explicitly show that $R (z) \neq 0$ whenever $z$ is inside the spectrum,
\begin{equation}
\Sigma^{(0)} := \sigma (\Delta_\mathbb{B}) = [-2 \sqrt{\kappa},2 \sqrt{\kappa}] .
\end{equation}
To this end, we recall that the Lyapunov exponent of the free Laplacian $\mathcal{L}^{(0)}(z)$, $z \in \R$, can be computed explicitly (see, for example, 
\cite[Section 16.3]{aw_book}),
\begin{equation}
\mathcal{L}^{(0)}(z) = \frac{1}{2} \log(\kappa) ~\mbox{, } z \in \Sigma^{(0)} ~\mbox{,}
\end{equation}
so that
\begin{equation}
\dfrac{\ud}{\ud z} \mathcal{L}^{(0)} (z)  = 0 ~\mbox{, } z \in \Sigma^{(0)} ~\mbox{.}
\end{equation}

To prove that the remainder term is nonvanishing for $\kappa \geq 2$, 
\begin{equation}\label{eq:remainder1}
R(z) = \mathcal{L}^{(0)}(z) - \int_{\Sigma^{(0)}} \log \vert E - z \vert ~\ud n(E) \neq 0 ~\mbox{,}
\end{equation}
we consider the Thouless-like term,
\begin{equation}
\mathcal{L}_T (z) := \int_{\Sigma^{(0)}} \log \vert E - z \vert ~\ud n(E) ~\mbox{,}
\end{equation}
and prove that for all $z \in \C^+$, one has
\begin{equation}
\dfrac{\ud}{\ud z} \mathcal{L}_T (z) = \dfrac{\ud}{\ud z}  \int_{\Sigma^{(0)}} \log \vert E - z \vert ~\ud n(E) \neq 0  ~\mbox{.}
\end{equation}
We will then show that the limit $z = E + i \eta \to \E$, for $E \in \Sigma^{(0)}$, exists and is also nonzero. 

The known formula for the density of states measure of $\Delta_\mathbb{B}$ (see, for example, \cite[Corollary 6.2]{aw_book}) leads to the expression
\beq\label{eq:le1}
\mathcal{L}_T (z) = \frac{\kappa + 1}{2 \pi} \int_{\Sigma^{(0)}} \frac{(4 \kappa - E^2)^{\frac{1}{2}}}{(\kappa+1)^2 - E^2} \log | E -z| ~\ud E ~\mbox{.}
\eeq
Computing the derivative, we obtain
\beq\label{eq:le2}
\dfrac{\ud}{\ud z} {\mathcal{L}}_T (z) = \frac{\kappa + 1}{2 \pi} \int_{\Sigma^{(0)}} 
f_\kappa(E) \frac{dE}{(E-z)} ~\mbox{,}
\eeq
where 
\beq\label{eq:le3}
f_\kappa(E) := \frac{(4 \kappa - E^2)^{\frac{1}{2}}}{(\kappa+1)^2 - E^2} ~\mbox{.}
\eeq
We change variables so that $E := 2 \sqrt{\kappa} \cos \theta$ and obtain
\beq\label{eq:le3}
f_\kappa(E) = \frac{ 2 \sqrt{\kappa} \sin \theta }{(\kappa + 1)^2 - 4 \kappa \cos^2 \theta } ~\mbox{.}
\eeq
Consequently, from \eqref{eq:le2}-\eqref{eq:le3}, we obtain
\beq\label{eq:le4}
\dfrac{\ud}{\ud z}{\mathcal{L}}_T (z) = - \frac{2 \kappa(\kappa + 1)}{2 \pi} \int_{0}^\pi 
 \left( \frac{ \sin^2 \theta }{(\kappa + 1)^2 - 4 \kappa \cos^2 \theta } \right) 
\frac{d \theta}{2 \sqrt{\kappa} \cos \theta - z} ~\mbox{.}
\eeq
We next change variables by letting $w := e^{i \theta}$. After some manipulations, we arrive at
\bea\label{eq:le5}
{\mathcal{L}}_T (z) & = & - \frac{i  \kappa(\kappa + 1)}{4 \pi} \int_{|w| =1} 
 \left( \frac{ (w^2 - 1)^2  }{w^2 (\kappa + 1)^2 - \kappa (w^2 + 1)^2 } \right) 
\frac{d w}{ \sqrt{\kappa} (w^2 + 1)  - zw} \nonumber  \\
 & = &  \frac{i  (\kappa + 1)}{4 \pi} \int_{|w| =1} 
 \left( \frac{ (w^2 - 1)^2  }{(w^2 - \kappa)(w^2 - \frac{1}{\kappa}) } \right) 
\frac{d w}{ \sqrt{\kappa} (w^2 + 1)  - zw} ~\mbox{.}
\eea

To analyze the integrand in \eqref{eq:le5}, we note that the first factor has four real, simple poles:
two poles at $\pm \sqrt{\kappa}$ outside the contour $|w| = 1$, and two poles inside the contour at $\pm \frac{1}{\sqrt{\kappa}}$.  
The second factor in the integrand of (\ref{eq:le5}) has two poles at
$$
w_\pm = \frac{z}{2 \sqrt{\kappa}} \pm \frac{1}{2 \sqrt{\kappa}} [ z^2 - 4 \kappa]^{\frac{1}{2}} ~\mbox{.}
$$
For $z \in \Sigma^{(0)}$, we have $|w_\pm| = 1$, so they lie on the contour.

\medskip

\subsubsection{Evaluation of the contour integral}

We first replace the contour $|w| = 1$ by a deformed contour $\Gamma_\epsilon$.
In the neighborhoods of $w_\pm$, we replace the path along the unit circle  by a 
small circular arc with $|w| \leq 1$ and $|w - w_\pm| = \epsilon \ll 1$. Then, there are two poles $\pm {\kappa}^{-\frac{1}{2}}$ inside the contour. Let 
$$
F_z(w) :=  \left( \frac{ (w^2 - 1)^2  }{(w^2 - \kappa)(w^2 - \frac{1}{\kappa}) } \right) 
\frac{1}{{\kappa}^{-\frac{1}{2}} (w^2 + 1)  - zw} ~\mbox{.}
$$
The residues of $F_z$ at poles $\pm  {\kappa}^{-\frac{1}{2}}$ are 
\beq\label{eq:contour1}
\rm{Res} ~~(F_z(w), \pm  {\kappa}^{-\frac{1}{2}}) = \pm  \frac{(1 - \kappa)}{2 (1+\kappa)(1+ \kappa \mp z)} ~\mbox{.}
\eeq
Consequently, by the Residue Theorem we obtain
\beq\label{eq:contour3}
\int_{\Gamma_\epsilon} ~F_z(w) ~dw =  2 \pi i \left(  \frac{1-\kappa}{1+\kappa}
\right) \left( \frac{z}{(1 + \kappa)^2 - z^2}\right) ~\mbox{.}
\eeq

\medskip
\subsubsection{Poles on the unit circle}

We calculate the contributions to the contour integral \eqref{eq:le5} coming from the two circular arcs near $w_\pm$. This requires calculating the residues of $F_z(w)$ at $w_\pm$. Since $|w| =1$, we write $w_\pm = e^{i \phi}$, for some $\phi \in [0, 2\pi)$. The first factor of $F_z(w)$ is regular at $w_\pm$, so we obtain
\beq\label{eq:residue1}
\left( \frac{ (w_\pm^2 - 1)^2  }{(w_\pm^2 - \kappa)(w_\pm^2 - \frac{1}{\kappa}) } \right) 
= \frac{-4 \sin^2 \phi}{2 \cos 2\phi - (\kappa + \kappa^{-1})} ~\mbox{.}
\eeq
 As for the singular factor, we get
 \bea\label{eq:residue2}
 \lim_{w \to w_\pm} (w - w_\pm) \left
( \frac{1}{ {\kappa}^{\frac{1}{2}} (w^2 + 1)  - zw} \right) & = & \lim_{w \to w_\pm} (w - w_\pm)  \frac{1}{\kappa^{\frac{1}{2}} (w - w_+)(w - w_-)} \nonumber \\
 & = &
 \frac{\pm 1}{ \kappa^{\frac{1}{2}} ( w_+ - w_-)}  =\frac{\pm 1}{2i \kappa^{\frac{1}{2}} \sin \phi} ~\mbox{.}
 \eea
 This shows that the contributions from the poles $w_\pm$ on the contour $|w| = 1$ cancel. 
 
 
 \subsubsection{Summary: Nonconstant derivative}
 Collecting the results in \eqref{eq:contour3}  and  \eqref{eq:residue2}, we find
 \beq\label{eq:der_le1}
\dfrac{\ud}{\ud z} {\mathcal{L}}_T (z) = \left( \frac{\kappa - 1}{2}  \right) \left(\frac{z}{(1+\kappa)^2 - z^2}   \right), {\rm for} ~~ z \in \Sigma^{(0)} ~\mbox{.}
 \eeq
 Upon integrating, we thus obtain that the remainder $R(z)$ is given by
 \beq\label{eq:lyap3}
 R(z) = \frac{1}{2} \log \kappa + \frac{(\kappa - 1)}{4} \log [ (\kappa + 1)^2 - z^2] + C_0(\kappa) ~\mbox{,}
 \eeq
 where $C_0(\kappa)$ is an integration constant. 
 When $\kappa = 1$, describing the Laplacian on $\Z$, the derivative on the right hand side of (\ref{eq:der_le1}) vanishes and $R(z) = 0$ consistent with the fact that the Lyapunov exponent is zero. 
 Since it is known that the Lyapunov exponent for $\Delta_{\mathbb{B}}$ is equal to $\frac{1}{2} \log \kappa$, the non-constant second term in (\ref{eq:lyap3}) proves that the remainder $R(z)$ in the modified Thouless formula (\ref{eq:lyap_thouless1}) is indeed non-trivial.


\begin{appendices}

\section{Appendix: Proof of Lemma \ref{lemma_action-maps}} \label{app_lemma}
\setcounter{equation}{0}
We consider the cases as itemized in the statement of Lemma \ref{lemma_action-maps}. Throughout the following argument, $d_1, \dots, d_\ell$ refer to the vertex $x=(0, a_1, \dots, a_\ell)$, as defined in (\ref{eq_prop_representation_formulas}). 
\begin{itemize}
\item[(i)] Suppose that $b_1 \neq \kappa$. Then, by definition in (\ref{eq_transl1}) and as emphasized in Remark \ref{remark_tau1}, the map $\tau_1$ acts as right shift. In particular, repeated application of the definitions of $\tau_1$ and $\tau_2$ yields 
\begin{align}
\tau_x(z) & = (\tau_2^{d_1} \tau_1) \dots (\tau_2^{d_\ell} \tau_1) (0, b_1, \dots , b_m) \\
&  = (\tau_2^{d_1} \tau_1) \dots  (\tau_2^{d_{\ell-1}} \tau_1) (0, d_\ell, (b_1 + d_\ell)_{\rm{mod}(\kappa)}, \dots , (b_m + d_\ell)_{\rm{mod}(\kappa)} ) \nonumber \\
 & = \dots = (\tau_2^{d_1} \tau_1) (0, d_2 , \dots , (\sum_{j=2}^{\ell - 1} d_j )_{\rm{mod}(\kappa)} ,  (b_1 + \sum_{j=2}^{\ell} d_j)_{\rm{mod}(\kappa)} , \dots , (b_m + \sum_{j=2}^{\ell} d_j)_{\rm{mod}(\kappa)}  ) ~\mbox{.} \label{eq_comp-lemma_case}
\end{align}
Here, we note that the definition in (\ref{eq_prop_representation_formulas}) implies that $d_2, \dots, d_\ell \neq \kappa$, so the assumption that $b_1 \neq \kappa$ ensures that at each of the $(\ell-1)$ iterative steps in (\ref{eq_comp-lemma_case}), the first component of the tuple (after $0$) is never equal to $\kappa$ (in particular, the operation modulo $(\kappa + 1)$ could be omitted) and, consequently, that $\tau_1$ keeps acting as a right shift. 

Applying the last product $(\tau_2^{d_1} \tau_1)$ in (\ref{eq_comp-lemma_case}) and using the telescoping property in (\ref{eq_compositional-lemma-key}), we thus conclude
\begin{align*}
\tau_x(z) & = (0, d_1, \dots , (\sum_{j=1}^{\ell - 1} d_j )_{\rm{mod}(\kappa)} , (b_1 + \sum_{j=1}^{\ell} d_j )_{\rm{mod}(\kappa)}, \dots , (b_m + \sum_{j=1}^{\ell} d_j )_{\rm{mod}(\kappa)} ) \\
   & = (0, a_1, \dots , a_{\ell - 1} , (b_1 + a_\ell)_{\rm{mod}(\kappa)}, \dots , (b_m + a_\ell)_{\rm{mod}(\kappa)} ) ~\mbox{.}
\end{align*}

\item[(ii)] If $b_1 = \kappa$ and $z = (0, b_1)$ (i.e., $m =1$), then $\tau_1$ acts according to the second equality in the definition (\ref{eq_transl1}), in particular
\begin{equation}
\tau_x(0, b_1) = (\tau_2^{d_1} \tau_1) \dots (\tau_2^{d_\ell} \tau_1) (0, \kappa) = (\tau_2^{d_1} \tau_1) \dots (\tau_2^{d_\ell-1} \tau_1) (0) ~\mbox{.}
\end{equation}
Since $d_1 , \dots d_\ell -1 \neq \kappa$, all subsequent actions of $\tau_1$ amount to right shifts, so that in analogy to case (i) above, we obtain
\begin{align*}
\tau_x(0, b_1) & = (\tau_2^{d_1} \tau_1) \dots (\tau_2^{d_\ell-1} \tau_1) (0) \nonumber \\
 & = (0, d_1 , (d_1 + d_2)_{\rm{mod}(\kappa)}, \dots , (\sum_{j=1}^{\ell - 1} d_j)_{\rm{mod}(\kappa)}) = (0, a_1, \dots , a_{\ell - 1}) ~\mbox{,}
\end{align*}
where the last equality again relies on the the telescoping property in (\ref{eq_compositional-lemma-key}).

\item[(iii)] Suppose that $b_1 = \kappa$ and $z = (0, b_1, \dots , b_m)$ with $m \geq 2$. Then, by the definition in (\ref{eq_transl1}) and as highlighted in Remark \ref{remark_tau1}, $\tau_1$ acts as a left shift and by adding $1$ in the new first entry (after 0), thus 
\begin{align}
\tau_x(z) & = (\tau_2^{d_1} \tau_1) \dots (\tau_2^{d_\ell} \tau_1) (0, \kappa, b_2, \dots , b_m) \nonumber \\
 & = (\tau_2^{d_1} \tau_1) \dots (\tau_2^{d_{\ell-1}} \tau_1) (0, (b_2 +1)_{\rm{mod}(\kappa+1)}, b_3, \dots b_m) ~\mbox{.} \label{eq_comp-lemma-caseiii-1}
\end{align}
In order to iteratively apply the products $(\tau_2^{d_j} \tau_1)$ in (\ref{eq_comp-lemma-caseiii-1}), $1 \leq j \leq \ell-1$, we suppose that $n$ is the {\em{smallest}} index with $1 \leq n \leq \min\{m , \ell\} - 1$ which satisfies the condition that
\begin{align}
c_{n+1} := \left[ \left(b_{n+1} + \sum_{j=\ell - (n-2)}^\ell d_j \right)_{\rm{mod}(\kappa)} + 1 + d_{\ell - (n-1)} \right]_{\rm{mod} (\kappa + 1)} \neq \kappa ~\mbox{;}
\end{align}
here, we use the convention that sums of the form $\sum_{j = \alpha}^{\beta} x_j$ with $\alpha > \beta$ evaluate to zero. Then, when iteratively applying the first $n$ products in (\ref{eq_comp-lemma-caseiii-1}) (reading from right to left), $\tau_1$ acts as a left shift according to the third equality in (\ref{eq_transl1}), which results in
\begin{align}
\tau_x(z) & = (\tau_2^{d_1} \tau_1) \dots (\tau_2^{d_{\ell-1}} \tau_1)(z) \nonumber \\ 
&= (\tau_2^{d_1} \tau_1) \dots (\tau_2^{d_{\ell - (n+2)}} \tau_1) (0, c_{n+1} , (b_{n+2} + \sum_{j= \ell - (n - 1)}^\ell d_j)_{\rm{mod}(\kappa)}, \dots , (b_m + \sum_{j= \ell - (n - 1)}^\ell d_j)_{\rm{mod}(\kappa)} ) \label{eq_eq_comp-lemma-caseiii-2}
\end{align}

Now, since by definition $c_{n+1} \neq \kappa$, we can use similar considerations than in case (i) to compute (\ref{eq_eq_comp-lemma-caseiii-2}), which, appealing to the telescoping property in (\ref{eq_compositional-lemma-key}), yields 
\begin{equation}
\tau_x(z) = (0, a_1, \dots, a_{\ell - n} , (c_{n+1} + a_{\ell - n})_{\rm{mod}(\kappa)}, (b_{n+2} + a_{\ell})_{\rm{mod}(\kappa)}, \dots, (b_{m} + a_{\ell})_{\rm{mod}(\kappa)} ~\mbox{.}
\end{equation}

\item[(iv)] For $m \geq 2$, suppose that $b_1 = \kappa$ and for all indices $n$ with $1 \leq n \leq \min\{m , \ell\}$, one has
\begin{align}
\left[ \left(b_{n+1} + \sum_{j=\ell - (n-2)}^\ell d_j \right)_{\rm{mod}(\kappa)} + 1 + d_{\ell - (n-1)} \right]_{\rm{mod} (\kappa + 1)} = \kappa ~\mbox{.}
\end{align}
Then, if $m \leq \ell$, splitting off the first $m$ products (reading from left to right) of $\tau_x$ 
\begin{align} \label{eq_comp-lemma-caseiv-1}
\tau_x(z) & = (\tau_2^{d_1} \tau_1) \dots (\tau_2^{d_{\ell-m}} \tau_1) \left\{ (\tau_2^{d_{\ell-m+1}} \tau_1) \dots (\tau_2^{d_\ell} \tau_1)  (0, b_1, \dots b_m) \right\} ~\mbox{,}
\end{align}
the map $\tau_1$ acts as a left shift according to the third equality in (\ref{eq_transl1}) in these first $m$ products; thus, (\ref{eq_comp-lemma-caseiv-1}) implies
\begin{align} 
\tau_x(z) & = (\tau_2^{d_1} \tau_1) \dots (\tau_2^{d_{\ell-m}} \tau_1) (0) = (0, a_1, \dots , a_{\ell-m}) ~\mbox{.}
\end{align}
Similarly, if $m > \ell$, $\tau_1$ in all factors of $\tau_x$ acts as a left shift, so that 
\begin{align*}
\tau_x(z) & = (\tau_2^{d_1} \tau_1) \dots (\tau_2^{d_\ell} \tau_1) (0, b_1, \dots , b_\ell, b_{\ell + 1} , \dots b_m) \nonumber \\
   & = (0, \left[ (b_{\ell + 1} + \sum_{j=2}^{\ell} d_j)_{\rm{mod}(\kappa)} + 1 + d_1 \right]_{\rm{mod}(\kappa+1)}, (b_{\ell+1} + \sum_{j=1}^{\ell} d_j)_{\rm{mod} (\kappa)} , \dots , (b_{m} + \sum_{j=1}^{\ell} d_j)_{\rm{mod} )(\kappa)} \nonumber \\
   & = (0, \left[(b_{\ell+1} +a_\ell - a_1)_{\rm{mod} (\kappa)} + 1 + a_1\right]_{\rm{mod}(\kappa+1)}, (b_{\ell+1} + a_\ell)_{\rm{mod} (\kappa)}, \dots, (b_{m} + a_\ell)_{\rm{mod} (\kappa)}) ~\mbox{.}
\end{align*}

\end{itemize}

\end{appendices}


\end{document}